\documentclass[english,11pt,a4paper]{article}
\usepackage[english]{babel}
\usepackage{latexsym,amssymb,amsmath,amsfonts,amscd,epsfig,amsthm,color,mathrsfs,mdframed,mleftright}
\usepackage{aliascnt}
\usepackage{graphicx}
\usepackage{tikz}
\usepackage{quantikz}
\usepackage{algorithm}
\usepackage{algpseudocode}
\usepackage{float}
\usepackage[left=2.3cm,right=2.3cm,top=2.3cm,bottom=2.3cm]{geometry}
\usepackage{ifthen}
\usepackage{authblk}
\usepackage{colonequals}
\usepackage{array}
\usepackage{pgfplots}
\usepgfplotslibrary{groupplots}
\pgfplotsset{compat=1.18}
\usetikzlibrary{decorations.pathreplacing}
\usepackage{comment}

\makeatletter
\newtheorem*{rep@theorem}{\rep@title}
\newcommand{\newreptheorem}[2]{%
\newenvironment{rep#1}[1]{%
 \def\rep@title{#2 \ref*{##1}}%
 \begin{rep@theorem}}%
 {\end{rep@theorem}}}
\makeatother

\usepackage{hyperref}
\hypersetup{
    bookmarksnumbered=true, 
    unicode=false, 
    pdfstartview={FitH}, 
    pdftitle={Optimal Quantum Algorithm for Ground-State Energy Estimation with a Guiding State}, 
    pdfauthor={Stacey Jeffery and Freek Witteveen}, 
    pdfsubject={}, 
    pdfcreator={}, 
    pdfproducer={}, 
    pdfkeywords={}, 
    pdfnewwindow=true, 
    colorlinks=true, 
    linkcolor=blue, 
    citecolor=blue, 
    filecolor=blue, 
    urlcolor=blue 
}

\newcommand{\SJ}[1]{}
\newcommand{\FW}[1]{}

\usepackage[capitalize,nameinlink]{cleveref}

\newtheorem{theorem}{Theorem}[section]

\newaliascnt{definition}{theorem}
\newtheorem{definition}[definition]{Definition}
\aliascntresetthe{definition}
\Crefname{definition}{Definition}{Definitions}

\newaliascnt{lemma}{theorem}
\newtheorem{lemma}[lemma]{Lemma}
\aliascntresetthe{lemma}
\Crefname{lemma}{Lemma}{Lemmas}

\newaliascnt{proposition}{theorem}

\aliascntresetthe{proposition}
\Crefname{proposition}{Proposition}{Propositions}

\newaliascnt{corollary}{theorem}
\newtheorem{corollary}[corollary]{Corollary}
\aliascntresetthe{corollary}
\Crefname{corollary}{Corollary}{Corollaries}

\newaliascnt{claim}{theorem}

\aliascntresetthe{claim}
\Crefname{claim}{Claim}{Claims}

\newaliascnt{example}{theorem}

\aliascntresetthe{example}
\Crefname{example}{Example}{Examples}

\newaliascnt{conjecture}{theorem}

\aliascntresetthe{conjecture}
\Crefname{conjecture}{Conjecture}{Conjectures}

\newaliascnt{aside}{theorem}

\aliascntresetthe{aside}
\Crefname{aside}{Aside}{Asides}

\newaliascnt{remark}{theorem}
\newtheorem{remark}[remark]{Remark}
\aliascntresetthe{remark}
\Crefname{remark}{Remark}{Remarks}

\def\bigO#1{O\mleft(#1\mright)}

\def\ket#1{{\lvert}#1\rangle}

\def\bra#1{{\langle}#1\rvert}
\def\braket#1#2{{{\langle}#1\vert}#2\rangle}

\def\ceil#1{{\lceil}#1\rceil}

\def\norm#1{\left\| #1 \right\|}

\newcommand{\eps}{\varepsilon}
\newcommand{\R}{\mathbb{R}}
\newcommand{\CC}{\mathbb{C}}

\newcommand{\NN}{\mathbb{N}}

\newcommand{\defeq}{\colonequals}
\newcommand{\ketbra}[2]{|{#1}\rangle\!\langle{#2}|}
\newcommand{\ot}{\otimes}
\newcommand{\op}{\oplus}

\title{Optimal Quantum Algorithm for\\ Ground-State Energy Estimation with a Guiding State}
\author[1,2]{Stacey Jeffery}
\author[1]{Freek Witteveen}
\affil[1]{QuSoft \& CWI, Amsterdam}
\affil[2]{QLever \& University of Amsterdam}

\date{}

\begin{document}

\maketitle

\vspace{-35pt}

\begin{abstract}
In the problem of ground-state energy estimation, one aims to estimate the smallest eigenvalue of a Hamiltonian, often given a guiding state, with some promised overlap $\gamma$ with the ground space. The main approach to this problem is to simulate its evolution, and estimate the smallest (or equivalently, largest) eigenphase of the resulting unitary $U$. We give a quantum algorithm that estimates the largest eigenphase of a unitary $U$ in this guided setting using a factor of $\log\frac{1}{\gamma}$ fewer queries to $U$ than the previous best approach. The result matches an existing lower bound, and answers an open question from Mande and de Wolf.  The algorithm is based on transducers, which often allow composition of quantum algorithms without overhead from error reduction.

\end{abstract}

\section{Introduction}

One of the most promising applications of quantum computers is the simulation of complex quantum systems that are relevant to physics and chemistry. A central task in this domain is to estimate the ground state energy of a Hamiltonian.
There are various approaches to solving this on a quantum computer. For fault-tolerant quantum computation, the standard paradigm for computing ground state energies is based on quantum phase estimation (QPE), \cite{kitaev1996PhaseEst,abrams1999quantum}, see \cite{bauer2020quantum,mcardle2020quantum} for a discussion of the application to computational physics and chemistry. Here, we consider the problem where we would like to know the ground state energy $E_0$ of a Hamiltonian $H$ to precision $\delta$, and we assume that we are given the ability to prepare a state $\ket{\psi}$ that has overlap $\gamma$ with the ground state. The computational problem is the \emph{guided ground state energy problem}:

\begin{mdframed}
Given a Hamiltonian $H$ with eigen-decomposition
\begin{align*}
    H = \sum_{k=0}^{K-1} E_k \Pi_k, \qquad E_0 < E_1 < \dots < E_{K-1}
\end{align*}
and a state-preparation unitary $A$ that prepares a state $\ket{\psi}$ such that $\norm{\Pi_0 \ket{\psi}} \geq \gamma > 0$, estimate $E_0$ to precision $\delta > 0$.
\end{mdframed}

Here, we have not specified how we have access to the Hamiltonian. What matters is that one can implement unitary time evolution on the quantum computer (which is possible in various access models to $H$), or encode the spectrum of $H$ in a different way in a unitary operator.
If we let $U = e^{-iH}$ be the time-evolution unitary, or alternatively if we let $U$ be the qubitized walk operator associated to $H$, the guided ground state energy problem directly reduces to the following \emph{max eigenphase estimation} problem:

\begin{mdframed}
Given access to a unitary $U$ with eigen-decomposition
\begin{align*}
    U = \sum_{k=0}^{K-1} e^{ix_k} \Pi_k, \qquad x_0 > x_1 > \dots > x_{K-1} > 0
\end{align*}
and a state-preparation unitary $A$ that prepares a state $\ket{\psi}$ such that $\norm{\Pi_0 \ket{\psi}} \geq \gamma > 0$, estimate $x_0$ to precision $\delta > 0$.
\end{mdframed}

A basic method to solve this problem would be to run quantum phase estimation with $\ket{\psi}$ as the initial state to precision $\delta$. Then, \emph{ignoring discretization error} one obtains an estimate of $x_k$ with probability $\norm{\Pi_k \ket{\psi}}^2$. This uses $\bigO{\delta^{-1}}$ calls to the unitary $U$. One can then repeat this $\bigO{\gamma^{-2}}$ times to obtain an estimate of the maximum eigenphase $x_k$. Instead of repeating, one can use amplitude amplification to reduce the scaling quadratically in $\gamma$. However, there is one subtlety: one cannot quite ignore the effect of discretization, which introduces an error in QPE. Previous work needed to reduce this error to be of the order of $\gamma$, causing a logarithmic overhead in $\gamma$, given a total complexity of $\bigO{\frac{1}{\gamma\delta}\log \frac{1}{\gamma}}$ \cite{mande2023tight,gilyen2018QSingValTransfThesis,Lin2020nearoptimalground,martyn2021grand}. On the other hand, \cite{mande2023tight} shows a query lower bound $\Omega\mleft(\frac{1}{\gamma \delta}\mright)$, leaving a logarithmic gap between upper and lower bounds, and leaves the precise asymptotic query complexity as an open question. Given the fundamental nature of the task, and its importance for practical applications of quantum computing, we find it of interest to close this gap, and this is the main result of this work. Specifically, we prove the following.
\begin{mdframed}
\begin{theorem}\label{thm:main-intro}
    There exists a quantum algorithm that solves the maximum eigenphase estimation problem (and thereby the guided ground state energy problem) to precision $\delta$ and with advice state overlap $\gamma$ using 
    \begin{center}
        $\bigO{\frac{1}{\gamma} \log \frac{1}{\delta}\log\log\frac{1}{\delta}}$ controlled calls to $A$, and $A^\dagger$, $\quad\bigO{\frac{1}{\gamma \delta}}$ controlled calls to $U$, and $U^\dagger$, 
    \end{center}
    and $\bigO{\frac{1}{\gamma \delta} \log \frac{1}{\gamma \delta}+\frac{1}{\gamma}\log N \log \frac{1}{\delta}\log\log\frac{1}{\delta}}$ one- and two-qubit gates and classical operations, where $N$ is the dimension of the system.
\end{theorem}
\end{mdframed}
This result is restated and proven as \cref{thm:main}.
This number of calls to $U$ is optimal, as it matches the lower bound from \cite{mande2023tight}, whereas the number of calls to $A$ is optimal up to the log factors.
The basis for the algorithm in \cref{thm:main-intro} is a quantum algorithm for a decision version of the task: given $U$ and $A$, and parameters $s$ and $\delta'$, decide between $x_0 \geq s$, or $x_0 \leq s - \delta'$. If one can solve this threshold problem, which we call the \emph{max eigenphase threshold problem} one can perform a binary search to determine $x_0$ to precision $\delta$. In most applications (e.g. when computing ground state energies of Hamiltonians), when compiling to a quantum circuit the dominant cost is the implementation of $U$. The state-preparation unitary can also be expensive to implement (but is used a factor $\delta$ less by the algorithm). The additional elementary gates and operations our algorithm uses will in most cases be a subleading contribution -- for example, $\log N$ can be assumed to be dominated by the cost of $A$, which generates a state on this many qubits. The algorithm in \cref{thm:main-intro} uses an additional $\bigO{\log \frac{1}{\gamma \delta}}$ qubits (apart from the system of dimension $N$ on which $U$ and $A$ act).

\subsection{Previous approaches}
Previous approaches required a $\log \frac{1}{\gamma}$ overhead in the complexity of computing max eigenphases.
We now describe two broad approaches to this problem, and explain why they lead to this logarithmic overhead.

First of all, there is the standard quantum phase estimation algorithm. Here, the issue is that when using $n$-bit phase estimation, there is a discretization error.
If the phase to be estimated does not have an exact representation with $n$ bits, then there is some probability of obtaining an estimate that is far from the true value. This cannot be avoided \cite{mande2023tight}. If one uses phase estimation as a subroutine to find the maximum eigenphase, one needs to suppress the error. The error needs to be smaller than the true signal of the maximum eigenphase, and this leads to a $\log \frac{1}{\gamma}$ overhead. Details can be found in \cite{mande2023tight}.

There is also a general approach to the problem based on quantum singular value transformations (QSVT)
\cite{gilyen2018QSingValTransfThesis,Lin2020nearoptimalground,martyn2021grand,dong2022ground}, which also suffers a $\log \frac{1}{\gamma}$-factor in the number of calls to the unitary. In this framework, the reason for the appearance of this factor is the following. QSVT is used to solve the max eigenphase threshold problem described above. To that end, it would like to implement an `ideal' function which behaves like a step function that projects onto all eigenspaces with $x_k \geq s$. However, it needs to approximate this by a polynomial. There are two parameters in this approximation: the width of the band in which the polynomial switches from being nearly zero to being nearly one, and the size of the ripples around the target values zero and one. To solve the energy threshold problem to precision $\delta$, the width of the transition should be at most $\delta$, but also the size of the ripples should not be larger than $\gamma$ (otherwise we cannot distinguish between a component of $\bigO{\gamma}$ in the target interval, and a $\bigO{1}$ component in its complement). To achieve such a polynomial approximation, a degree of $\Theta(\delta^{-1} \log\frac{1}{\gamma})$ is necessary and sufficient.

These challenges are illustrated in \cref{fig:qpe-log-overhead}.
Both for the standard phase estimation and for the QSVT framework, it was not known how to get rid of the factor $\log \frac{1}{\gamma}$, and we believe that prior to this work it was unclear whether one could avoid it at all.

\begin{figure}[t]
\centering
\begin{tikzpicture}
\begin{groupplot}[
    group style={group size=2 by 1, horizontal sep=1.6cm},
    width=0.47\textwidth,
    height=6.0cm,
    axis lines=left,
    clip=false,
    tick label style={font=\small},
    label style={font=\small},
    title style={font=\small},
]

\nextgroupplot[
    ymin=0, ymax=0.75,
    xmin=-0.02, xmax=1.02,
    xlabel={output estimate},
    ylabel={probability},
    xtick={0,0.25,0.5,0.75,1},
    xticklabels={$0$,$1/4$,$1/2$,$3/4$,$1$},
    ytick={0,0.2,0.4,0.6},
    title={(a) Standard QPE},
]

\addplot+[
    ybar,
    bar width=1.5pt,
    draw=black,
    fill=black!30
] coordinates {
    (0.0000,0.003906)
    (0.0625,0.005183)
    (0.1250,0.007905)
    (0.1875,0.014976)
    (0.2500,0.043735)
    (0.3125,0.684895)
    (0.3750,0.171959)
    (0.4375,0.028355)
    (0.5000,0.011719)
    (0.5625,0.006739)
    (0.6250,0.004655)
    (0.6875,0.003642)
    (0.7500,0.003140)
    (0.8125,0.002942)
    (0.8750,0.002980)
    (0.9375,0.003267)
};

\addplot[dashed, thick] coordinates {(0.3333,0) (0.3333,0.73)};
\node[anchor=south] at (axis cs:0.3333,0.73) {$x$};

\nextgroupplot[
    ymin=-0.18, ymax=1.08,
    xmin=-1, xmax=1,
    xlabel={eigenphase},
    ylabel={$p(x)$},
    ytick={0,0.08,0.92,1},
    yticklabels={$0$,$\gamma$,$1-\gamma$,$1$},
    title={(b) QSVT threshold polynomial},
    declare function={
        g=0.08;      
        d=0.46;      
        w=18;        
        t(\x)=(\x + d/2)/d;
        smoothstep(\x)=3*t(\x)^2 - 2*t(\x)^3;
        leftbranch(\x)=g/2*(1 - cos(deg(w*(\x + d/2))));
        rightbranch(\x)=1 - g/2*(1 - cos(deg(w*(\x - d/2))));
        poly(\x)= (\x < -d/2) ? leftbranch(\x)
                 : ((\x > d/2) ? rightbranch(\x)
                 : smoothstep(\x));
    },
]

\addplot[
    draw=none,
    fill=black!10
] coordinates {
    (-0.23,-0.18)
    ( 0.23,-0.18)
    ( 0.23, 1.08)
    (-0.23, 1.08)
} -- cycle;

\addplot[densely dotted, thin] coordinates {(-1,0.08) (1,0.08)};
\addplot[densely dotted, thin] coordinates {(-1,0.92) (1,0.92)};

\addplot[dashed, very thick] coordinates {
    (-1,0) (0,0)
    (0,1)
    (1,1)
};

\addplot[solid, very thick, domain=-1:1, samples=300] {poly(x)};

\draw[<->, thick]
    (axis cs:-0.23,-0.07) -- (axis cs:0.23,-0.07);
\node at (axis cs:0,-0.135) {$\delta$};

\end{groupplot}
\end{tikzpicture}
\caption{
An illustration of why previous approaches incur a $\log(1/\gamma)$ overhead.
\textbf{(a)} In $n$-bit quantum phase estimation, if the true phase $x$ does not have an exact $n$-bit binary representation, then the measurement distribution is spread over multiple outcomes and has nonzero tails. Suppressing the probability of a bad estimate below the signal level $\gamma$ leads to logarithmic overhead.
\textbf{(b)} In a QSVT-based approach, one approximates an ideal step function by a polynomial. The transition width must be at most $\delta$, while the ripple size away from the transition must be at most $\gamma$. Achieving both requires degree $\Theta(\delta^{-1}\log(1/\gamma))$. 
}
\label{fig:qpe-log-overhead}
\end{figure}
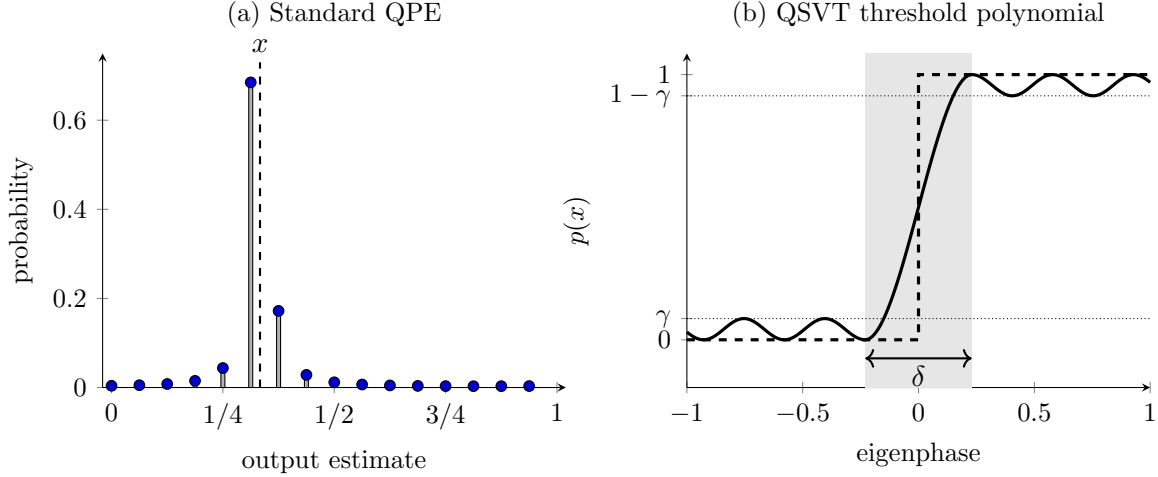

\subsection{Methods: transducers for phase estimation}

Transducers are a framework for designing and analyzing quantum algorithms, developed in \cite{belovs2024transducers} by combining and building on~\cite{belovs2023LasVegas,jeffery2022subroutines}; see also \cite{belovs2024purifier,jeffery2026quantum,apers2026elfs,chen2026time} for applications. Transducers are a model of quantum computation that, in particular, allows for the composition of quantum subroutines in a way that efficiently handles variable-complexity subroutines and error reduction for quantum subroutines. We give a brief review of transducers in \cref{sec:transducers-prelim}. For now, we just remark that transducers are useful for getting rid of (typically logarithmic) overheads that arise from error reduction. Indeed, suppose one has a quantum subroutine with bounded error, and suppose this subroutine is used $n$ times in a quantum algorithm. To ensure that the algorithm has constant probability of success, naively one needs to perform error reduction on the subroutine so it has probability of error at most $\bigO{\frac{1}{n^2}}$, which ensures that the $\bigO{\frac{1}{n}}$ operator norm errors induced by each imperfect subroutine call adds up to a constant. This error reduction comes at overhead $\bigO{\log n}$.

However, transducers provide an alternative approach to composing subroutines that avoids this overhead. Here, a bounded-error quantum algorithm is modeled by a different object (the transducer) that composes without error, and after composition one can convert back to a quantum algorithm with bounded error. In fact, given a quantum algorithm ${\cal A}$ that outputs some $z\in [q]$ such that there exists $z^*\in [q]$ whose output probability is at least $\frac{1}{2}+\delta$, there exists a transducer, called a \emph{purifier} that uses $O(\frac{1}{\delta})$ calls to ${\cal A}$ to compute $z^*$ \emph{exactly} in the model of transducers. An exact, error-free transducer, when converted to a quantum circuit, will have bounded error, so this is not a recipe for removing error. However, if ${\cal A}$ will be called $n$ times as a subroutine, as described above, we have removed the need for $O(\log\frac{1}{n})$ repetitions, thus saving a $O(\log\frac{1}{n})$ factor. A natural approach then to remove the $\log\frac{1}{\gamma}$ factor in max phase estimation is to purify phase estimation to get a (exact) transducer and compose it with a transducer for amplitude estimation. This does not quite work, because the phase estimation algorithm does not have the property that there exists one estimate $z^*$ that is output with probability $1/2+\delta$. 

However, we are nonetheless able to design an exact transducer for a decision version of phase estimation: given access to a unitary $U$, an eigenstate $\ket{\psi}$ with eigenphase $x$ (so $U\ket{\psi} = e^{ix} \ket{\psi}$), it decides whether $x < s$ or $x>s$ for some threshold value $s$. In addition, we give a simple transducer for a decision version of amplitude amplification: given a state $\ket{\psi}$ and a projection onto a marked subspace $\Pi_M$, it decides whether $\norm{\Pi_M\ket{\psi}} \geq \gamma$, or $\Pi_M \ket{\psi} = 0$. 
Composing these two transducers gives a transducer for the threshold max phase estimation problem: We again consider a unitary $U$ and a state $\ket{\psi}$ (that now need not be an eigenstate), and we let $\Pi_{> t}$, $\Pi_{< t}$ denote projections onto the eigenspaces with eigenphases above or below $t$ respectively. The max eigenphase threshold problem decides between $\norm{\Pi_{>s} \ket{\psi}} \geq \gamma$, or $\Pi_{>s - \delta} \ket{\psi} = 0$, for some parameters $s, \gamma, \delta > 0$.
This transducer gives a quantum algorithm solving the max eigenphase threshold problem using $\bigO{\frac{1}{\gamma \delta}}$ queries, which we can use to do a binary search to obtain \cref{thm:main-intro}.

The transducers we describe for (decision) amplitude amplification (in \Cref{sec:decision-aa-transducer} and (decision) phase estimation (in \Cref{sec:decision-phase-est-transducer}) may be of interest in their own right. While any quantum algorithm can be mapped to a transducer, this generic mapping has overheads; for example, the most general results require quantum random access gates in order to perform different parts of the program in superposition. Here we avoid the need for quantum random access gates by finding the transducers that exploit the problem structure.

\subsection{Related Work}

In \cite{mande2023tight}, it was shown that the number of queries to $U$ needed to solve the maximum eigenphase estimation problem with success probability $1 - \eps$ is $\Omega(\frac{1}{\gamma\delta}\log\frac{1}{\eps})$, as long as one of the three parameters, $\eps$, $\delta$ or $\gamma$, is constant. Simultaneously with our work, \cite{somma2026lowerbound} showed that this lower bound holds in all three variables, as long as $N\geq \frac{1}{\gamma^2}\log\frac{1}{\eps}$. Our \Cref{thm:main-intro} implies an upper bound of $\bigO{\frac{1}{\gamma\delta}\log\frac{1}{\eps}}$ on the number of calls to $U$, in the $\eps$-error setting: simply run the algorithm from \Cref{thm:main-intro} $O(\log\frac{1}{\eps})$ times and take the median. Thus, our work, combined with \cite{somma2026lowerbound} shows that the number of queries to $U$ needed to solve maximum eigenphase estimation with error $\eps$ is $\Theta(\frac{1}{\gamma\delta}\log\frac{1}{\eps})$, when $N$ is sufficiently large. Ref.~\cite{somma2026lowerbound} also shows that this lower bound holds in the special case where there is a unique ground state, and a spectral gap at least $\delta$. 

\subsection{Outlook}
We derive a query-optimal algorithm for the guided ground state energy problem (or the guided max eigenphase problem) in terms of queries to $U$. The improvement compared to previous work is logarithmic in one over the ground state overlap $\gamma$. An interesting question for future work is whether this leads to faster algorithms for parameter and problem settings relevant to applications in physics and chemistry. We have not attempted to optimize the constant factors in the algorithm design and analysis. However, the algorithm we describe is not particularly complicated, and does not introduce large constant overheads, so it is not unlikely to be competitive with state-of-the-art approaches and resource estimates. 

The approach in this work has a space overhead of $\log\frac{1}{\gamma \delta}$ qubits, while the standard approach only needs $\log\frac{1}{\delta}$ extra qubits for phase estimation and amplitude amplification. Can the transducer approach be modified to also further reduce the space complexity?

Since quantum phase estimation is a central subroutine in many algorithms, we expect that using transducers, it is possible to remove similar logarithmic factors in other algorithms as well. One interesting aspect of this is the question to what extent the algorithm can be made fully coherent, where it is known that there are some limitations if one does not make any assumptions on the eigenphases (see \cite{rall2021faster} and the discussion of a rounding promise for making phase estimation fully coherent). We leave a systematic study to future work.

\subsection{Acknowledgments: Colleagues, Funding and AI}

We thank Rolando D.~Somma and Ronald de Wolf for helpful comments, sharing an early draft of~\cite{somma2026lowerbound}, and coordinating arXiv submission. 

This work is co-funded by the European Union (ERC, ASC-Q, 101040624); and Divide \& Quantum  (with project number 1389.20.241) of the research programme NWA-ORC, which is (partly) financed by the Dutch Research Council (NWO). 

LLMs were used in this work for some computations, proofs, and figures. All high-level ideas came from the authors. We have checked all proof details, and have rewritten all LLM-generated proofs to meet our standard of writing.

\section{Preliminaries}

We let $\mathbb{N}_0$ be the set of non-negative integers, and $\mathbb{N}_1$ the set of positive integers.

\subsection{Max eigenphase estimation}

We consider a unitary $U$ acting on $\CC^N$ with eigen-decomposition
\begin{equation}\label{eq:eigendecomp}
    U=\sum_{k=0}^{K-1}e^{ix_k}\Pi_k, \qquad x_0>x_1>\dots>x_{K-1},
\end{equation}
and an advice-preparation unitary $A$ acting on $\CC^N\otimes \CC^{N'}$ (the second register
holds possible garbage) with
\begin{equation}\label{eq:advice}
    A\ket{0}=\ket{\psi}=\sum_{k=0}^{K-1}\alpha_k\ket{\psi_k},
    \qquad (\Pi_k\otimes I)\ket{\psi_k}=\ket{\psi_k},\ \norm{\ket{\psi_k}}=1.
\end{equation}
Note that $U\otimes I$ has the same eigenphases as $U$, with eigenprojectors $\Pi_k\otimes I$, so without loss of generality, 
we assume $N'=1$. 
For
$t\in\R$ we write $\Pi_{>t}=\sum_{k:x_k>t}\Pi_k$ and $\Pi_{\geq t}$ analogously.
 
\begin{definition}[Window promise]\label{def:window}
    We say the pair $(U,A)$ satisfies the \emph{window promise} if all eigenphases $x_k$ of $U$
    lie in $[0,\pi/2]$.
\end{definition}

We will assume this property for $U$ for the remainder of this work. It is straightforward to ensure in the application of ground state energy estimation, and it holds for the lower bound constructions in \cite{mande2023tight,somma2026lowerbound}.
 
\begin{definition}[Threshold max eigenphase problem]\label{def:threshold}
    Fix $s\in(0,\pi/2]$, $\delta'\in(0,s)$ and $\gamma\in(0,1]$. Given $(U,A)$ as above satisfying
    the window promise, decide between:
    \begin{description}
        \item[Positive case:] $\norm{\Pi_{>s}\ket{\psi}}\geq \gamma$;
        \item[Negative case:] $\norm{\Pi_{>s-\delta'}\ket{\psi}}=0$.
    \end{description}
\end{definition}
 
\begin{definition}[Guided max eigenphase estimation problem, following {\cite[Definition~2.3]{mande2023tight}}]\label{def:estimation}
    Fix $\gamma\in(0,1]$ and $\delta\in(0,1/8]$. Given $(U,A)$ satisfying the window promise and
    the advice promise
    \begin{equation}\label{eq:advice-promise}
        \norm{\Pi_{\geq \theta_{\max}}\ket{\psi}}^2 = \sum_{k:\,x_k=\theta_{\max}}|\alpha_k|^2 = |\alpha_0|^2\geq
        \gamma^2, \qquad \theta_{\max}\defeq x_0,
    \end{equation}
    output $\hat\theta$ such that $|\hat\theta-\theta_{\max}|\leq \delta$.
\end{definition}

A closely related problem is to estimate ground state energies of some Hamiltonian.
Here, we assume we are given a Hamiltonian $H$ with eigen-decomposition
\begin{align}
    H = \sum_{k=0}^{K-1} E_k \Pi_k, 
    \qquad E_0 < E_1 < \dots < E_{K-1}
\end{align}
and again an advice-preparation unitary $A$. We assume that we can implement time evolution $U(t) = e^{-itH}$ on a quantum computer. We further assume without loss of generality (since we can always shift $H$ by the identity) that $E_{K-1} \leq 0$, and we also assume we have a bound $\lambda>0$ such that $E_{0} \geq -\lambda$.

\begin{definition}[Guided ground state energy estimation problem]\label{def:ground-energy}
    Fix $\gamma\in(0,1]$ and $\delta\in(0,1/8]$. Given $(H,A)$ and
    the advice promise
    \begin{equation}\label{eq:advice-promise-gs}
        \norm{\Pi_0 \ket{\psi}}^2 \geq
        \gamma^2,
    \end{equation}
    output $\hat E_0$ such that $|\hat E_0-E_0|\leq \delta$.
\end{definition}

This is a special case of \cref{def:estimation}: choose $t = \pi/(2\lambda)$, then $U = e^{-itH}$ satisfies the window promise, and the max eigenphase translates to the ground state energy. Note that the rescaling by $\lambda$ means that \cref{thm:main-intro} implies that we can solve the guided ground state energy problem using time evolution for time $\bigO{\lambda/(\delta \gamma)}$ if we want to estimate $E_0$ to precision $\delta$.

\subsection{Transducers}\label{sec:transducers-prelim}

Let $S$ be a unitary on a direct sum of spaces, ${\cal H}\oplus {\cal L}$, which we call the \emph{public space} (or boundary) and \emph{private space}, respectively. Then there exists a unitary $U_S$ such that for any $\ket{\xi}\in {\cal H}$, there exists $\ket{v}\in {\cal L}$ such that 
$$S(\ket{\xi}+\ket{v})=U_S\ket{\xi}+\ket{v},$$
and moreover, for any $\ket{\xi}\in {\cal H}$ and $\ket{v}\in {\cal L}$, if $S(\ket{\xi}+\ket{v})=\ket{\tau}+\ket{v}$ for some $\ket{\tau}\in {\cal H}$, then $\ket{\tau}=U_S\ket{\xi}$~\cite[Theorem 3.1]{belovs2024transducers}. The vector $\ket{v}$ is called a \emph{catalyst}, and it depends on $\ket{\xi}$. It is not usually a unit vector, even when $\ket{\xi}$ is. 

The above claim is not obvious, and what is also non-obvious is that, using sufficiently many calls to $S$, it is always possible to approximate $U_S$, so we should think of $U_S$ as being something we \emph{want} to implement, and $S$ -- the \emph{transducer} -- as something we \emph{can} implement. One example of a transducer is a \emph{quantum walk operator} for a search problem (as in \cite{szegedy2004QMarkovChainSearch} or \cite{belovs2013ElectricWalks}), which turns out to be a transducer with \emph{transduction action} 
$$U_S:\ket{\sigma}\mapsto \left\{
\begin{array}{ll}
\ket{\sigma} & \mbox{if }M\neq \emptyset\\
-\ket{\sigma} & \mbox{if }M=\emptyset,
\end{array}
\right.$$
where $\ket{\sigma}$ is the initial distribution, and $M$ is the marked set (see \cite[Section 6]{belovs2024transducers} or \cite{belovs2024weldedTrees} for more details). Generally, as in this example, many applications of $S$ are needed to approximate $U_S$.
We alternatively write $U_S:\ket{\xi}\mapsto \ket{\tau}$ as $S:\ket{\xi}\rightsquigarrow \ket{\tau}$. 

Transducers are a model of quantum computation with the following basic features: 
\begin{enumerate}
    \item Any transducer can be mapped to a bounded-error quantum algorithm (see \Cref{thm:transducer-implementation-informal} or \Cref{thm:transducer-implementation}).
    \item Any bounded-error quantum algorithm can be mapped to a transducer (in multiple ways, in fact). We will not use this feature, as we will design transducers directly.
    \item These mappings preserve complexity in a certain useful sense.
\end{enumerate}
The main point of transducers is that they compose much more nicely than bounded-error quantum algorithms, as we will see in an explicit example in \Cref{sec:transducers}. A related feature is that they can be mapped to quantum algorithms that are surprisingly query-efficient, as we explain in \Cref{rem:query}.

\paragraph{Complexities of transducers.} There are multiple notions of complexity associated with a transducer $S$. The most straightforward is its \emph{iteration time} $T(S)$, which is the complexity of implementing the unitary $S$. Note that this is a model-dependent quantity, and it is precisely where any specific gate model can be plugged into the transducer model. In our case, $T(S)$ will be the number of one- and two-qubit gates, plus some additional operations we will make explicit, needed to implement the unitary $S$, but one could redefine it to count, for example, only non-Clifford gates, or depth.

We will shortly see, in \Cref{thm:transducer-implementation-informal} and \Cref{thm:transducer-implementation}, that one can design a quantum algorithm for $U_S$ that makes some number of calls to $S$ -- these being the dominating cost. Then, in addition to $T(S)$, which captures the cost of one call to $S$, the number of calls to $S$ needed by such an algorithm is critical to understanding its complexity. This number is related to a complexity called the \emph{transduction complexity}, which is defined, for a specific input $\xi$, as $W(S,\xi)=\norm{\ket{v}}^2$, where $\ket{v}$ is the smallest\footnote{Catalysts are not unique, and $W$ is defined with respect to the smallest catalyst, but any catalyst gives an upper bound, which is generally all that's needed.} catalyst such that $S(\ket{\xi}+\ket{v})=U_S\ket{\xi}+\ket{v}$. Loosely speaking, if $W(S)$ is an upper bound on $W(S,\xi)$ over all possible inputs $\xi$ (this need not range over all unit vectors, if it is known the input will be from a smaller set) then $U_S$ can be approximated in complexity $O(T(S)W(S))$. In particular, we have the following informal version of \cite[Theorem~3.2]{belovs2024transducers}:
\begin{theorem}[Informal]\label{thm:transducer-implementation-informal}
    For a transducer $S$ and error parameter $\eps$, there is a quantum algorithm that implements $U_S$ with error $\eps$ using $O(1+W(S)/\eps^2)$ applications of $S$, and other basic gates.  
\end{theorem}
The algorithm is sketched in \cref{fig:transducer-to-algo}. For details on why the algorithm works -- which is not at all obvious, despite the algorithm being rather simple -- refer to~\cite{belovs2024transducers}.

\paragraph{Transducers with oracles, and canonical form.} We can also define a transducer with respect to an oracle $O$, which can be an arbitrary unitary acting on the same space as $S=S(O)$ (possibly trivially on all but a small subset of the qubits of $S$). One reason to make the use of an oracle explicit (rather than just letting it be part of the implementation of $S$) is if we want to compose transducers, in which case, $O$ can provide a boundary through which to do this. Another reason is that we might hope to get an algorithm that applies $O$ \emph{fewer} times than $S$ -- note that this would not be possible by simply using \Cref{thm:transducer-implementation-informal}, since each application of $S=S(O)$ would require calling $S$ some (non-negative) integer number of times. 

To understand how the number of queries needed per application of $S$ could be less than 1, we can define the \emph{Las Vegas query complexity}, which first arose in~\cite{belovs2023LasVegas}. To do this, we first assume that $S(O)$ is in \emph{canonical form}, which means that:
\begin{enumerate}
    \item $S(O)=S^\circ O$, for some oracle-independent unitary $S^\circ$ called the \emph{work unitary};
    \item the private space can be decomposed ${\cal L}={\cal L}^\bullet\oplus {\cal L}^\circ$, and $O$ acts trivially on ${\cal H}$ and ${\cal L}^\circ$. 
\end{enumerate}
For any input $\ket{\xi}\in {\cal H}$, we can define the Las Vegas query complexity on input $\ket{\xi}$ with respect to $O$ as:
$$L(S,O,\xi)=\norm{\Pi_{{\cal L}^\bullet}\ket{v}}^2,$$
where $\ket{v}$ is the (smallest possible) catalyst for $\xi$. 
This is the size of only the part of the catalyst on which $O$ acts non-trivially. As with $W$, we can let $L(S,O)$ denote the maximum $L(S,O,\xi)$ over all possible inputs $\xi$ -- meaning all of those for which we want the algorithm to work. 

\begin{remark}\label{rem:query}
As already noted, if we apply \Cref{thm:transducer-implementation-informal} to $S(O)$, then the number of queries to $O$ will be the number of queries $S(O)$ makes to $O$, times the number of calls to $S$ -- it will never be less than $O(W(S)/\eps^2)$ unless it's 0. For a canonical transducer, we can improve this situation with the following theorem, which follows from Theorem~3.3 of \cite{belovs2024transducers}.
\end{remark}

\begin{theorem}\label{thm:transducer-implementation}
    Let $S(O)=S^\circ O$ be a canonical form transducer. 
Let $\eps$, $W$ and $L$ be parameters.
    There exists an algorithm ${\cal A}(S)$ that acts on ${\cal H}\oplus {\cal L}$, and a $O(\log W+\log\frac{1}{\eps})$-qubit ancilla, and uses:
    \begin{enumerate}
        \item $\bigO{L/\eps^2}$ controlled calls to $O$
        \item $\bigO{1+W/\eps^2}$ controlled calls to $S^\circ$ and additional one- and two-qubit gates \emph{on the ancilla};
        \end{enumerate} and has the following correctness guarantee.
    On input $\ket{\xi}\in {\cal H}$, the algorithm outputs a state $\ket{\tilde\tau}$ such that $\norm{U_S\ket{\xi}-\ket{\tilde\tau}}\leq \eps$ as long as $W(S,O,\xi)\leq W$ and $L(S,O,\xi)\leq L$. 
\end{theorem}
Informally, if we choose $L$ and $W$ appropriately, we get an $\eps$-error algorithm for $U_S$ that makes $L(S,O)$ calls to $O$ and $W(S,O)$ calls to $S^{\circ}$ and additional gates.
The algorithm is similar to the one illustrated in \cref{fig:transducer-to-algo}, except that while one applies the work unitary in every step, the oracle is only called every $W/L$ steps.

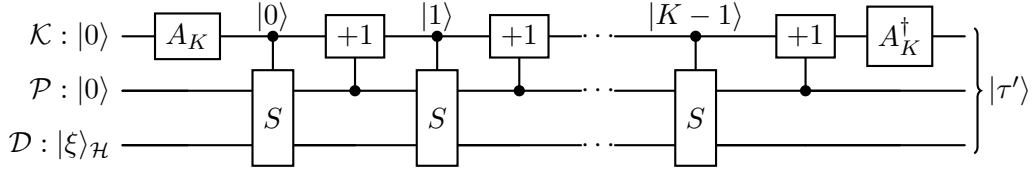
\begin{figure}
    \centering
    \begin{quantikz}[
  row sep={0.72cm,between origins},
  column sep=0.43cm
]
\lstick{$\mathcal K:\ket{0}$}
  & \gate{A_K}
  & \ctrl[
      style={label=above:{\ensuremath{\lvert 0\rangle}}}
    ]{1}
  & \gate{+1}
  & \ctrl[
      style={label=above:{\ensuremath{\lvert 1\rangle}}}
    ]{1}
  & \gate{+1}
  & \push{\cdots}
  & \ctrl[
      style={label=above:{\ensuremath{\lvert K-1\rangle}}}
    ]{1}
  & \gate{+1}
  & \gate{A_K^\dagger}
  & \rstick[3]{$\ket{\tau'}$}
  \\
\lstick{$\mathcal P:\ket{0}$}
  & \qw
  & \gate[wires=2]{S}
  & \ctrl{-1}
  & \gate[wires=2]{S}
  & \ctrl{-1}
  & \push{\cdots}
  & \gate[wires=2]{S}
  & \ctrl{-1}
  & \qw
  &
  \\
\lstick{$\mathcal D:\ket{\xi}_{\mathcal H}$}
  & \qw
  &
  & \qw
  &
  & \qw
  & \push{\cdots}
  &
  & \qw
  & \qw
  &
\end{quantikz}
    \caption{The quantum algorithm in \cref{thm:transducer-implementation-informal} that approximately implements $U_S$. Here ${\cal K}$ is a register with basis $\ket{0}, \dots, \ket{K-1}$ where $K$ is the number of calls the algorithm makes to $S$, and ${\cal P}$ is a qubit register that controls whether ${\cal D}$ is interpreted as ${\cal H}$ or ${\cal L}$ (so ${\cal H} \op {\cal L} \subseteq {\cal P} \ot {\cal D}$). $A_K$ is a unitary such that $A_K\ket{0} = 1/\sqrt{K} \sum_{k=0}^{K-1} \ket{k}$, and the operation labeled $+1$ increments the register ${\cal K}$ modulo $K$ (this is done in a specific way that avoids an overhead of $O(\log K)$ gates). See \cite{belovs2024transducers} for details.}
    \label{fig:transducer-to-algo}
\end{figure}

Any transducer can be transformed into canonical form in a way that preserves its complexities, but in this paper, it will actually be rather easy to make a canonical form transducer directly, so we will not use this ability. 

\subsection{Transducers with an infinite counter}

It will be convenient to construct transducers that use an infinite-dimensional Hilbert space, whose main function is to act as some kind of infinite counter. This is related to the interpretation of transducers as being a Las Vegas model for quantum algorithms (where classical randomized Las Vegas algorithms are algorithms that have no fixed bound on their running time but rather a bound on the expected running time). On an infinite-dimensional Hilbert space ${\cal C}$ with basis $\{\ket{\ell} \, : \, \ell \in  \NN_0\}$, we can apply the increment operator, 
\begin{equation}
    {\sf Inc}:=\sum_{\ell=0}^\infty \ket{\ell+1}\bra{\ell}.\label{eq:Inc}
\end{equation}
This is an isometry, rather than a unitary, and in a standard gate model, it would require infinitely many gates to implement (although the model that arises when one allows it is in some ways reasonable~\cite{jeffery2026quantum}). We will use it as a placeholder where it makes certain constructions work exactly, and whenever we want to turn these constructions into actual algorithms, we will replace it with a truncated version:
\begin{equation}
    {\sf Inc}^{[D]}:=\sum_{\ell=0}^{D-1} \ket{\ell+1}\bra{\ell},\label{eq:Inc-D}
\end{equation}
for some $D\in \mathbb{N}_1$, where addition is modulo $D$. The unitary ${\sf Inc}^{[D]}$ can be implemented using $O(\log D)$ one- and two-qubit gates. In the remainder of this section, we prove that in the context of \Cref{thm:transducer-implementation}, we can replace ${\sf Inc}$ with ${\sf Inc}^{[D]}$ for some sufficiently large $D$, without changing the behaviour of ${\cal A}(S)$ (\Cref{lem:truncate-no-error}). This is similar to \cite[Proposition~4.3]{belovs2024purifier}, but we try to prove a somewhat general statement that might be useful for similar arguments in the future. The idea is that the algorithm will start in a state with no support on $\ell>0$, and each application of $S$ will only increase $\ell$ by at most a constant, which means that the algorithm makes $K$ calls to $S$, we can truncate $\ell$ at $\bigO{K}$ without incurring error in the algorithm.

For the remainder of this section, we consider transducers with finite-dimensional public space ${\cal H}$, and infinite-dimensional ${\cal L}=\mathrm{span}\{\ket{\ell}_{\cal C}:\ell\in \mathbb{N}_0\}\otimes {\cal V}$ for some finite-dimensional ${\cal V}$. We call $\ell$ the \emph{level}. We write ${\cal L}_{\leq D} = \mathrm{span}\{\ket{\ell}_{\cal C}:\ell\in \{0, \dots, D\}\}\otimes {\cal V}$ for the restriction of the private space to level at most $D$. We consider a condition under which we can truncate ${\cal C}$ to finite dimensions. The existence of a good truncation is captured by the following.

\begin{definition}\label{def:perfect-truncation}
    Let $D \in \NN_1$ and $m \in \NN_0$. Consider a unitary map $S$ on ${\cal H} \op {\cal L}$, and a unitary map $S^{[D]}$ on ${\cal H} \op {\cal L}_{\leq D}$.
    Then we say $S^{[D]}$ is an exact $(D,m)$-truncation of $S$, if it satisfies:
    \begin{enumerate}
        \item $S$ maps ${\cal H} \op {\cal L}_{\leq \ell}$ into ${\cal H} \op {\cal L}_{\leq \ell+m}$ for all $\ell < D - m$,
        \item $S^{[D]}$ equals $S$ when restricted to ${\cal H} \op {\cal L}_{\leq D-m}$.
    \end{enumerate}
\end{definition}

The idea is that if we do not apply $S$ too often in an algorithm, starting from ${\cal H}$, the algorithm does not see the difference between $S$ and $S^{[D]}$ and we can truncate without error.

\begin{lemma}\label{lem:truncate-no-error}
    Let $D \in \NN_1$ and $m \in \NN_0$. Consider a transducer $S$, and an exact $(D,m)$-truncation $S^{[D]}$ of $S$.
    Then, if ${\cal A}(S)$ and ${\cal A}(S^{[D]})$ are the algorithms constructed in \cref{thm:transducer-implementation-informal} using $K < D / m$ calls to the transducer, these two algorithms act identically on every state $\ket{\xi} \in {\cal H}$.
    Similarly, if we have a canonical transducer, and $S^\circ$ and $O$ have exact $(D,m_S)$- and $(D,m_O)$-truncations $S^{\circ, [D]}$ and $O^{[D]}$, then if ${\cal A}(S)$ and ${\cal A}(S^{[D]})$ are the algorithms constructed in \cref{thm:transducer-implementation} using $K_S$ calls to the work unitary and $K_O$ calls to the oracle such that $m_S K_S + m_O K_O < D$, then these two algorithms act identically on every state $\ket{\xi} \in {\cal H}$. 
\end{lemma}

\begin{proof}    
    If $G_0,\dots,G_K$ are maps that do not change the level, such as those that only control on $\mathrm{span}\{\ket{\ell}_{\cal C}:\ell\in \mathbb{N}_0\}$, and act arbitrarily on other registers, then the assumption on $S$ and $S^{[D]}$ directly implies that $G_K S G_{K-1} S \dots G_1 S G_0$ equals $G_K S^{[D]} G_{K-1} S^{[D]} \dots G_1 S^{[D]} G_0$ on any input not supported on ${\cal L}$ as long as $Km < D$.
    The algorithm in \cref{thm:transducer-implementation-informal} is of this form, and the argument for \cref{thm:transducer-implementation} is identical, when one notes that the algorithm referred to in \Cref{thm:transducer-implementation} also starts in $\ket{0}\otimes \ket{\xi}$, and consists entirely of (1) maps\footnote{In fact, just as in the algorithm from \Cref{thm:transducer-implementation-informal} shown in \Cref{fig:transducer-to-algo}, these maps are simply $A_K$ and $A_K^\dagger$ at the beginning and end, and controlled increments (sometimes by more than 1) of ${\cal K}$.} that act as the identity on $\mathrm{span}\{\ket{\ell}_{\cal C}:\ell\in \NN_0\}$ or are controlled by it; (2) controlled calls to $S^\circ$; and (3) controlled calls to $O$ (see Section 7.3 of \cite{belovs2024transducers}).
\end{proof}

This applies in particular to the case where the infinite register is used by the increment operator ${\sf Inc}$. E.g. if $S$ only calls ${\sf Inc}$ once (and all other operations act diagonally on $\ell$), then we can truncate the transducer at level $D = K+1$ if we make $K$ calls to $S$, without incurring any truncation error.

\section{Transducers}\label{sec:transducers}

The main result of this section is a canonical transducer for the threshold max eigenphase problem (\Cref{def:threshold}), which is built from transducers for decision versions of amplitude amplification and phase estimation.

Before we state the main result of this section, consider the two input oracles that are relevant to the problem: $A$ and $U$. Suppose we designed a transducer $S$ that could be implemented using $O(1)$ calls to $A$ and $U$, and with $W(S)=O(\frac{1}{\gamma\delta})$. We need not treat this as a transducer with an oracle, or worry about whether it is in canonical form to get an algorithm. By simply applying \Cref{thm:transducer-implementation-informal}, we get a quantum algorithm for threshold max eigenphase that uses $O(\frac{1}{\gamma\delta})$ applications of $S$, and thus, this many queries to each of $A$ and $U$. We cannot hope to improve in the number of queries to $U$, but we can hope to reduce the number of queries to $A$. In order to achieve this, we could make a canonical transducer $S(A)$ with respect to the oracle $A$, in which $L(S,A)\leq O(\frac{1}{\gamma})$. Then by \Cref{thm:transducer-implementation}, we would get a quantum algorithm that only makes $O(\frac{1}{\gamma})$ queries to $A$. We will not do this, but we will do something just as good. Define:
\begin{equation}\label{eq:O-A}
O_\gamma(A):=I-2\ket{\varphi_\gamma(A)}\bra{\varphi_\gamma(A)},\quad\mbox{where }\ket{\varphi_\gamma(A)}:=\sqrt{\frac{\gamma}{1+\gamma}}\ket{0}\ket{0}+\frac{1}{\sqrt{1+\gamma}}\ket{1}A\ket{0}.
\end{equation}
We will exhibit a transducer that is canonical with respect to this oracle. Since this oracle can be implemented using one call to each of $A$ and $A^\dagger$ (see \Cref{lem:A-impl} for details), $O(\frac{1}{\gamma})$ calls to this oracle translates to $O(\frac{1}{\gamma})$ calls to $A$ and its adjoint. 
Formally, in this section, we prove the following.
\begin{theorem}\label{thm:main-transducer}
Fix any $s\in (0,\pi/2]$, $\delta\in (0,s)$ and $\gamma\in (0,1]$. 
There exists a canonical transducer $S(O)=S^\circ O$ that depends on an input $(U,A)$ to the threshold max eigenphase problem as follows: $O=O_\gamma(A)$, and $S^\circ$ depends on $U$ (but not $A$); and such that the following hold:
\begin{enumerate}
    \item ${\cal H}=\mathrm{span}\{\ket{\bar{1}}\ket{0}\}$, and ${\cal L}=\mathrm{span}\{\ket{\ell}:\ell\in \mathbb{N}_0\}\otimes \mathbb{C}^N$ is infinite dimensional.\label{item:spaces}
    \item If $(U,A)$ is a positive instance of the threshold max eigenphase problem, $S:\ket{\bar{1}}\ket{0}\rightsquigarrow\ket{\bar{1}}\ket{0}$.\label{item:pos}
    \item If $(U,A)$ is a negative instance of the threshold max eigenphase problem, $S:\ket{\bar{1}}\ket{0}\rightsquigarrow -\ket{\bar{1}}\ket{0}$.\label{item:neg}
    \item If $(U,A)$ is either a positive or negative instance of the threshold max eigenphase problem, $W(S,O,\ket{\bar{1},0})\leq \bigO{\frac{1}{\gamma\delta}}$ and $L(S,O,\ket{\bar{1},0})\leq 1+ \frac{1}{\gamma}$.\label{item:complexities}
    \item The oracle is an exact $(D,1)$-truncation of itself; and there exists an exact $(D,3)$-truncation of the work unitary that can be implemented using $O(1)$ controlled calls to $U$ and $U^\dagger$, and $O(\log D)$ additional one- and two-qubit gates.\label{item:more-spaces}
\end{enumerate}
\end{theorem}

In brief: the theorem states that there exists a transducer that decides the threshold max eigenphase problem, using a constant number of controlled calls to $U$, $U^\dagger$, and with an oracle that makes a constant number of calls to $A$ and $A^\dagger$. \Cref{item:complexities} implies that it yields an algorithm that uses $\bigO{\frac{1}{\gamma\delta}}$ controlled calls to $U$ and $U^\dagger$ and $\bigO{\frac{1}{\gamma}}$ controlled calls to $A$ and $A^\dagger$. \Cref{item:more-spaces} is sufficient to allow us to later apply \cref{lem:truncate-no-error} to truncate the infinite-dimensional register.
We prove this theorem in \Cref{sec:direct-construction} by describing and analyzing the claimed transducer.
Note that the transducer implements the transformation $$S:\ket{\bar{1}}\ket{0}\rightsquigarrow \pm \ket{\bar{1}}\ket{0},$$ where the sign encodes whether we have a positive or negative instance. To turn this into a algorithm for distinguishing the two cases, we can use single-bit phase estimation, see also \cite{belovs2024purifier}.

As a warm-up, we first give transducers for a decision version of amplitude amplification (\Cref{sec:decision-aa-transducer}), and for a decision version of phase estimation  (\Cref{sec:decision-phase-est-transducer}). These could then be composed using the results of \cite{belovs2024transducers} to prove \Cref{thm:main-transducer}, but in this case, it is somewhat more instructive to simply describe and analyze a composed transducer directly, which is what we do in \Cref{sec:direct-construction}.

\subsection{Transducer for decision amplitude amplification}\label{sec:decision-aa-transducer}

This section gives a transducer for the following decision version of amplitude amplification. 
\begin{definition}[Decision amplitude amplification]\label{def:decision-aa}
Given a state preparation unitary $A$ that acts on $\mathbb{C}^N$ as $A\ket{0}=\ket{\psi}$, and a reflection $R_M=2\Pi_M-I$, decide if $\Pi_M\ket{\psi}=0$.    
\end{definition}

For unit vectors $\ket{\psi(M)}\in M$ and $\ket{\psi({\overline{M}})}\in M^\bot$, write:
$$\ket{\psi}=\sqrt{\eps}\ket{\psi(M)}+\sqrt{1-\eps}\ket{\psi({\overline{M}})},$$
so in particular, $\eps=\norm{\Pi_M\ket{\psi}}^2$.
Let $w$ be a small positive weight, to be fixed later. 
The transducer can be visualized as a quantum walk on the following graph. 

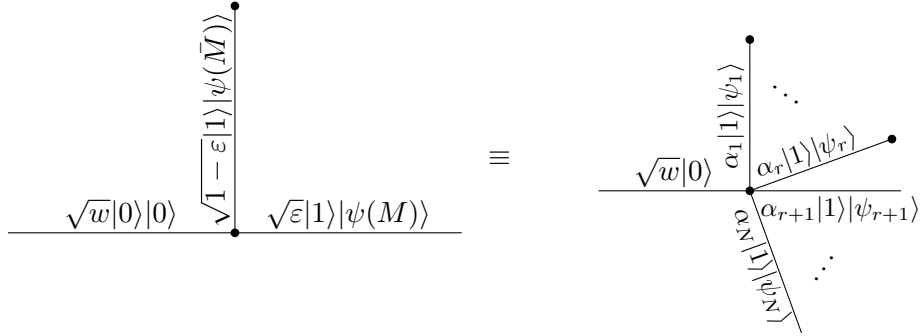
\begin{figure}[h]
\begin{center}
\begin{tikzpicture}
    \node at (0,0) {\begin{tikzpicture}
    \draw (-2,0) -- (4,0);
    \draw (1,0) -- (1,3);
    \filldraw (1,0) circle (.05);
    \filldraw (1,3) circle (.05); 

    \node at (-.5,.25) {$\sqrt{w}\ket{0}\ket{0}$};
    \node at (2.5,.25) {$\sqrt{\eps}\ket{1}\ket{\psi(M)}$};
    \node[rotate=90] at (.75,1.5) {$\sqrt{1-\eps}\ket{1}\ket{\psi({\bar{M}})}$};
\end{tikzpicture}};
\node at (3.5,-.5) {$\equiv$};
\node at (7,-.9) {\begin{tikzpicture}
    \def\angA{-70}   
    \def\angB{20}   
    \def\R{2}       

    \coordinate (c) at (1,0);

    \draw (-1,0) -- (3,0);
    \draw (c) -- (1,2);
    \draw (c) -- ++(\angA:\R) coordinate (a);
    \draw (c) -- ++(\angB:\R) coordinate (b);

    \filldraw (c) circle (.05);
    \filldraw (1,2) circle (.05);
    \filldraw (b) circle (.05);

    \node at (0,.25) {$\sqrt{w}\ket{0}$};
    \node[rotate=90] at (.75,1) {\small $\alpha_1\ket{1}\ket{\psi_1}$};
    \node[rotate=-35] at (1.5,1.25) {$\dots$};
    \node[rotate=\angB] at (1.75,.5) {\small $\alpha_r\ket{1}\ket{\psi_r}$};
    \node at (2.2,-.25) {\small $\alpha_{r+1}\ket{1}\ket{\psi_{r+1}}$};
    \node[rotate=55] at (2,-1) {$\dots$};
    \node[rotate=\angA] at (1.15,-1) {\small $\alpha_N\ket{1}\ket{\psi_N}$};
\end{tikzpicture}};
\end{tikzpicture}
\end{center}
\caption{Visualization of the transducer for decision amplitude amplification. We will use such visualizations throughout this paper, for intuition, though all transducers will be described precisely in math. Each circle represents a vertex, which corresponds to a state -- often consisting of the sum of states labeling its incident edges -- or set of states (equivalently a space). The reflection around this space will be part of the transducer. 
In particular, all graphs we consider will be bipartite, and the transducer will be a product of reflections: one around the spaces corresponding to the first part of the bipartition, and another around the spaces corresponding to the other part of the bipartition. Note that some edges have only one endpoint. We call these boundary edges.}\label{fig:aa-transducer}
\end{figure}
In the right-hand side formulation of \Cref{fig:aa-transducer}, $\ket{\psi_{r+1}},\dots,\ket{\psi_N}$ is an orthonormal basis for $M$, and $\ket{\psi_1},\dots,\ket{\psi_r}$ is an orthonormal basis for its orthogonal complement $M^\bot$, and $\alpha_j=\braket{\psi_j}{\psi}$.  The right-hand formulation will be useful for understanding \Cref{sec:direct-construction}, but for now, the left-hand formulation is the most useful for intuition. Looking at this left-hand side, the intuition should be that a walker enters from the state $\ket{0}\ket{0}$, and does a random walk. Assume $w> 0$ is set so that we are unlikely to go back out the $\ket{0}\ket{0}$ edge again. Then the walker is essentially repeatedly choosing between the two other edges, with relative weights $\eps$ and $1-\eps$. When it goes up the weight-$(1-\eps)$ edge, it hits a dead end and goes back. If $\eps>0$, it will do this some number of times before eventually choosing the less likely weight-$\eps$ edge, which is an ``exit''.  If $\eps=0$, it will never leave by the weight-$\eps$ exit, so eventually, no matter what we set $w>0$, we will exit by the edge $\ket{0}\ket{0}$. 

We formally define the transducer now. The transducer will be the walk operator on the graph in \Cref{fig:aa-transducer}, which is a product of two reflections, $R_{\cal A}$, and $R_{\cal B}$, which we now define.\footnote{To understand how a walk operator is defined from a graph, see~\cite[Section 3.2.1]{jeffery2022kDist} or \cite{jeffery2026QWlecture}. Neither source is the originator of the idea of the quantum walk operator,  
but these most closely match our notation.}
First, $R_{\cal A}$ is defined as the reflection around the following space:
\begin{equation}\label{eq:R-A}
{\cal A}=\mathrm{span}\{\sqrt{w}\ket{0}\ket{0}+\sqrt{\eps}\ket{1}\ket{\psi(M)}+\sqrt{1-\eps}\ket{1}\ket{\psi({\bar{M}})}\}=\mathrm{span}\{\sqrt{w}\ket{0}\ket{0}+\ket{1}\ket{\psi}\},
\end{equation}
which we can efficiently reflect around using one call to each of $A$ and $A^\dagger$. In fact, note that this is simply $-O_\gamma(A)$ from \Cref{eq:O-A} with $\gamma=w$.

Looking at the left-hand side of \Cref{fig:aa-transducer}, we would need $R_{\cal B}$ to reflect around the span of $\ket{1}\ket{\psi({\overline{M}})}$, but reflecting around just this one state is actually not something we can do efficiently. However, since we will always stay in 
$$\mathrm{span}\{\ket{0}\ket{0},\ket{1}\ket{\psi(M)},\ket{1}\ket{\psi({{M}^\bot})}\},$$
it's sufficient to set $R_{\cal B}=(2\ket{1}\bra{1}\otimes \Pi_{\overline{M}}-I)=-\ket{1}\bra{1}\otimes R_M-\ket{0}\bra{0}\otimes I$, where $R_M$ is the reflection around $M$. Then this is precisely implementing the quantum walk on the right-hand side of \Cref{fig:aa-transducer}. 

Define the transducer by $S=R_{\cal B}R_{\cal A}$, with public and private spaces:
\begin{align*}
    {\cal H} &=\mathrm{span}\{\ket{0}\ket{0}\}\\
    {\cal L} &=\mathrm{span}\{\ket{1}\}\otimes \mathbb{C}^N.
\end{align*}

\begin{lemma}
    The transducer $S$ has transduction action 
    $$\ket{0}\ket{0}\rightsquigarrow\left\{\begin{array}{ll}
        \ket{0}\ket{0} & \mbox{if }\eps\neq 0 \\
        -\ket{0}\ket{0} & \mbox{if }\eps=0.
    \end{array}\right.$$
    If $\eps\neq 0$, $W(S)\leq\frac{w}{\eps}$, and if $\eps=0$, $W(S)\leq \frac{1}{w}$.
\end{lemma}
\begin{proof}
Suppose $\eps=0$. Letting $\ket{v}=\frac{1}{\sqrt{w}}\ket{1}\ket{\psi}$, we have
$$\ket{0}\ket{0}+\ket{v}=\frac{1}{\sqrt{w}}\left(\sqrt{w}\ket{0}\ket{0}+\ket{1}\ket{\psi}\right)\overset{R_{\cal A}}{\mapsto}\frac{1}{\sqrt{w}}\left(\sqrt{w}\ket{0}\ket{0}+\ket{1}\ket{\psi}\right)\overset{R_{\cal B}}{\mapsto}-\ket{0}\ket{0}+\frac{1}{\sqrt{w}}\ket{1}\ket{\psi},$$
since $\eps=0$ implies $-R_M\ket{\psi}=\ket{\psi}$. Then $W(S)=W(S,\ket{0,0})\leq \norm{\ket{v}}^2=\frac{1}{w}$, as claimed.

    On the other hand, suppose $\eps\neq 0$. Letting $\ket{v}=-\sqrt{\frac{w}{\eps}}\ket{1}\ket{\psi(M)}$, we have:
\begin{align*}
    (\sqrt{w}\bra{0}\bra{0}+\bra{1}\bra{\psi})(\ket{0}\ket{0}+\ket{v})&=\sqrt{w}+\bra{1}\bra{\psi}\left(-\sqrt{\frac{w}{\eps}}\ket{1}\ket{\psi(M)}\right)=\sqrt{w}-\sqrt{w}\frac{\braket{\psi}{\psi(M)}}{\sqrt{\eps}}=0,
\end{align*}
using $\ket{\psi(M)}=\frac{\Pi_M\ket{\psi}}{\norm{\Pi_M\ket{\psi}}}$ and $\norm{\Pi_M\ket{\psi}}^2=\eps$. Thus $R_{\cal A}$, reflects this state, giving:
\begin{align*}
    \ket{0}\ket{0}+\ket{v}
    &\overset{R_{\cal A}}{\mapsto} 
-\ket{0}\ket{0}+\sqrt{\frac{w}{\eps}}\ket{1}\ket{\psi(M)}\\
&\overset{R_{\cal B}}{\mapsto} 
\ket{0}\ket{0}-\sqrt{\frac{w}{\eps}}\ket{1}\ket{\psi(M)}=
\ket{0}\ket{0}+\ket{v}.
\end{align*}
Then $W(S)=W(S,\ket{0,0})\leq \norm{\ket{v}}^2=\frac{w}{\eps}$, as claimed.
\end{proof}

We have thus exhibited a transducer for decision amplitude amplification. As we should expect, as $\eps$ approaches 0 from above, $W(S)$ goes to infinity, so with no lower bound on $\eps$, we also have no upper bound on $W(S)$. However, if we have a lower bound, $\sqrt{\eps}\geq \gamma$, as in the case of \Cref{def:threshold}, then we can set $w=\gamma$ to get $W(S)\leq \frac{1}{\gamma}$ in both the 0 and 1 cases.

\subsection{Transducer for decision phase estimation}\label{sec:decision-phase-est-transducer}

We now give an infinite-dimensional transducer for \emph{threshold phase estimation}, which we define as follows. 
\begin{definition}[Decision Phase Estimation]\label{def:decision-phase-est}
    Given a threshold $\tilde s\in (0,\pi/2]$, oracle access to a unitary $U$ on $\mathbb{C}^N$, and an input eigenstate $\ket{\psi}$ such that $U\ket{\psi}=e^{ix}\ket{\psi}$ for $x\in[0,\pi/2]$, decide between:
    $x>\tilde s$ (positive case) and $x<\tilde s$ (negative case).
\end{definition}
\noindent As usual, we write $U=\sum_{k=0}^{K-1}e^{ix_k}\Pi_k$, and $\Pi_k\ket{\psi}=\alpha_k\ket{\psi_k}$ for a unit vector $\ket{\psi_k}$. 
For each $k$, define
$$\sigma=\frac{\pi}{2}-\tilde s,\;\theta_k=x_k+\sigma\mbox{ and }
\gamma_k=\frac{\sin\frac{\theta_k}{2}}{\cos\frac{\theta_k}{2}}.$$

\paragraph{The transducer.}
The transducer will be based on a quantum walk on a weighted infinite-dimensional line, similar to the purifier construction in~\cite{belovs2024purifier}. It will act on an infinite-dimensional register ${\cal C}$, in addition to the space acted on by $U$:
$$\mathrm{span}\{\ket{\ell}_{\cal C}:\ell\in\mathbb{N}_0\}\otimes \mathbb{C}^N.$$
In fact (as \Cref{fig:decision-phase-graph} suggests), we will not actually need the $\ell=0$ part of this space, but we still define the transducer on this space, in order to reuse parts of it more easily in the next section. The transducer will act as the identity on states with $\ell=0$ anyway, so we will simply put them in the public space and then never use them.

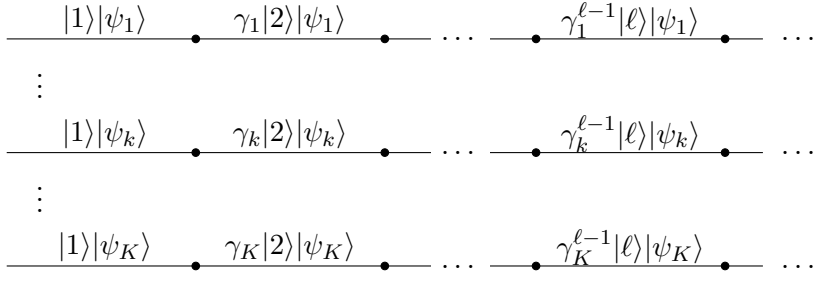
\begin{figure}
    \centering
\begin{tikzpicture}
\node at (0,3) {\begin{tikzpicture}
        \draw (0,0)--(10,0);
    \filldraw (2.5,0) circle (.05);
    \filldraw (5,0) circle (.05);
    \filldraw (7,0) circle (.05);
    \filldraw (9.5,0) circle (.05);

    \node at (1.3,.25) {$\ket{1}\ket{\psi_1}$};
    \node at (3.75,.25) {$\gamma_1\ket{2}\ket{\psi_1}$};

    \node[fill=white] at (6,0) {$\dots$};
    \node at (8.25,.25) {$\gamma_1^{\ell-1}\ket{\ell}\ket{\psi_1}$};
    \node at (10.5,0) {$\dots$};
\end{tikzpicture}};
\node at (-5,2.25) {$\vdots$};
\node at (0,1.5) {\begin{tikzpicture}
        \draw (0,0)--(10,0);
    \filldraw (2.5,0) circle (.05);
    \filldraw (5,0) circle (.05);
    \filldraw (7,0) circle (.05);
    \filldraw (9.5,0) circle (.05);

    \node at (1.3,.25) {$\ket{1}\ket{\psi_k}$};
    \node at (3.75,.25) {$\gamma_k\ket{2}\ket{\psi_k}$};

    \node[fill=white] at (6,0) {$\dots$};
    \node at (8.25,.25) {$\gamma_k^{\ell-1}\ket{\ell}\ket{\psi_k}$};
    \node at (10.5,0) {$\dots$};
\end{tikzpicture}};
\node at (-5,.75) {$\vdots$};
\node at (0,0) {\begin{tikzpicture}
        \draw (0,0)--(10,0);
    \filldraw (2.5,0) circle (.05);
    \filldraw (5,0) circle (.05);
    \filldraw (7,0) circle (.05);
    \filldraw (9.5,0) circle (.05);

    \node at (1.3,.25) {$\ket{1}\ket{\psi_K}$};
    \node at (3.75,.25) {$\gamma_K\ket{2}\ket{\psi_K}$};

    \node[fill=white] at (6,0) {$\dots$};
    \node at (8.25,.25) {$\gamma_K^{\ell-1}\ket{\ell}\ket{\psi_K}$};
    \node at (10.5,0) {$\dots$};
\end{tikzpicture}};
\end{tikzpicture}
\caption{Visualization of the transducer for decision phase estimation.}\label{fig:decision-phase-graph}
\end{figure}

The transducer will be a quantum walk operator for the graph in \Cref{fig:decision-phase-graph}, which is a collection of infinite lines, with one for each $k$. 
Concretely, define the transducer as the product of two reflections, $R_{\cal A}=R_1$ and $R_{\cal B}=R_0$, as follows:
\begin{equation}
\begin{split}
    R_1 &=2\sum_{k=0}^{K-1}\sum_{\ell=1}^\infty\left(\cos\frac{\theta_k}{2}\ket{2\ell-1}-i\sin\frac{\theta_k}{2}\ket{2\ell}\right)\left(\cos\frac{\theta_k}{2}\bra{2\ell-1}+i\sin\frac{\theta_k}{2}\bra{2\ell}\right)\otimes\Pi_k - I\\
    R_0 &= 2\sum_{k=0}^{K-1}\sum_{\ell=1}^\infty\left(\cos\frac{\theta_k}{2}\ket{2\ell}-i\sin\frac{\theta_k}{2}\ket{2\ell+1}\right)\left(\cos\frac{\theta_k}{2}\bra{2\ell}+i\sin\frac{\theta_k}{2}\bra{2\ell+1}\right)\otimes\Pi_k - I\\
    S &= R_1R_0.
    \end{split}\label{eq:S-decision-phase}
\end{equation}
To complete the definition of the transducer, define its public space and private space:
\begin{align*}
    {\cal H} &=\mathrm{span}\{\ket{0},\ket{1}\}\otimes \mathbb{C}^N\\
    {\cal L} &=\mathrm{span}\{\ket{2},\ket{3},\dots\}\otimes \mathbb{C}^N.
\end{align*}
As mentioned above, the $\ket{0}\otimes \mathbb{C}^N$ part of the public space is not really needed in this section, so we should think of the ``relevant'' public space as $\ket{1}\otimes\mathbb{C}^N$.

\paragraph{Implementation of the reflections.} 
To implement the transducer, it is sufficient to implement $R_0$ and $R_1$.
Interpreting ${\cal C}$ as two registers: a qubit register ${\cal C}_{lsb}$ encoding the least significant bit of $\ell$, and an infinite-dimensional register, ${\cal C}_{rest}$ encoding the remaining bits, so that, for example, $\ket{2\ell}=\ket{\ell}_{{\cal C}_{rest}}\ket{0}_{{\cal C}_{lsb}}$ and $\ket{2\ell+1}=\ket{\ell}_{{\cal C}_{rest}}\ket{1}_{{\cal C}_{lsb}}$, we have:
\begin{align*}
    R_0 &= 2\sum_{\ell=1}^{\infty} \ket{\ell}\bra{\ell}_{{\cal C}_{rest}}\sum_{k=0}^{K-1}\underbrace{\left(\cos\frac{\theta_k}{2}\ket{0}-i\sin\frac{\theta_k}{2}\ket{1}\right)}_{=:\ket{\phi_k}}\left(\cos\frac{\theta_k}{2}\bra{0}+i\sin\frac{\theta_k}{2}\bra{1}\right)_{{\cal C}_{lsb}}  \otimes\Pi_k - I\\
    &= -\ket{0}\bra{0}_{{\cal C}} \ot I + \sum_{\ell=1}^{\infty} \ket{\ell}\bra{\ell}_{{\cal C}_{rest}}\otimes \left(2\sum_{k=0}^{K-1}\ket{\phi_k}\bra{\phi_k}_{{\cal C}_{lsb}}\otimes \Pi_k-I\right).
\end{align*}

\begin{lemma}
    Let $R(U)=\ket{-}\bra{+}\otimes e^{i\sigma}U+\ket{+}\bra{-}\otimes e^{-i\sigma}U^\dagger$.
    Then $R(U)=2\sum_{k=0}^{K-1}\ket{\phi_k}\bra{\phi_k}\otimes \Pi_k-I$.
\end{lemma}
\begin{proof}
First, it is simple to verify that
$$\ket{\phi_k}=e^{-i\theta_k/2}\frac{1}{\sqrt{2}}\left(\ket{+}+e^{i\theta_k}\ket{-}\right),$$
from which it follows that:
\begin{equation}
    \begin{split}
        2\ket{\phi_k}\bra{\phi_k} &= \left(\ket{+}+e^{i\theta_k}\ket{-}\right)\left(\bra{+}+e^{-i\theta_k}\bra{-}\right)\\
        &= e^{i\theta_k}\ket{-}\bra{+}+e^{-i\theta_k}\ket{+}\bra{-}+I.
    \end{split}\label{eq:projector-form}
\end{equation}
We conclude by computing:
\begin{align*}
    R(U) &= \ket{-}\bra{+}\otimes\sum_{k=0}^{K-1}e^{i(\sigma+x_k)}\Pi_k+\ket{+}\bra{-}\otimes\sum_{k=0}^{K-1}e^{-i(\sigma+x_k)}\Pi_k\\
    &= \sum_{k=0}^{K-1}\left(e^{i\theta_k}\ket{-}\bra{+}+e^{-i\theta_k}\ket{+}\bra{-}\right)\otimes\Pi_k & \mbox{since }\theta_k=\sigma+x_k\\
    &=\sum_{k=0}^{K-1}(2\ket{\phi_k}\bra{\phi_k}-I)\otimes\Pi_k & \mbox{by \Cref{eq:projector-form}}
\end{align*}
from which the claimed result follows.
\end{proof}

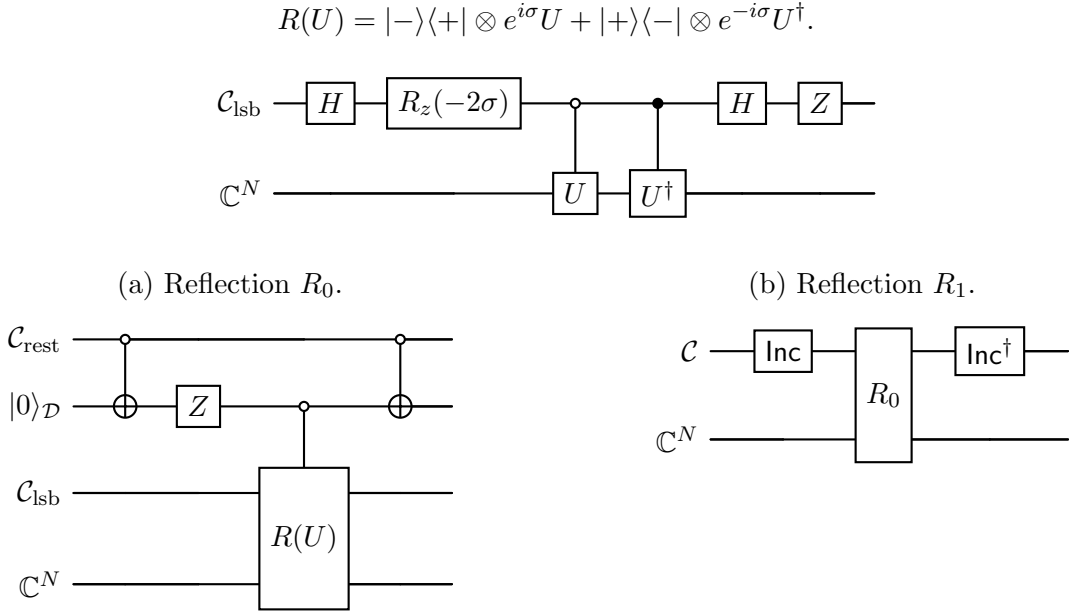
\begin{figure}[t]
\centering
\begin{minipage}{0.98\textwidth}
\centering
\[
R(U)=\ket{-}\bra{+}\otimes e^{i\sigma}U
     +\ket{+}\bra{-}\otimes e^{-i\sigma}U^\dagger .
\]
\begin{quantikz}[row sep=0.55cm,column sep=0.43cm]
\lstick{${\cal C}_{\mathrm{lsb}}$}
  & \gate{H}
  & \gate{R_z(-2\sigma)}
  & \octrl{1}
  & \ctrl{1}
  & \gate{H}
  & \gate{Z}
  & \qw \\
\lstick{$\mathbb C^N$}
  & \qw
  & \qw
  & \gate{U}
  & \gate{U^\dagger}
  & \qw
  & \qw
  & \qw
\end{quantikz}

\vspace{0.55cm}
\begin{minipage}[t]{0.48\linewidth}
\centering
(a) Reflection $R_0$.

\vspace{0.2cm}
\begin{quantikz}[row sep=0.54cm,column sep=0.52cm]
\lstick{${\cal C}_{\mathrm{rest}}$}
  & \octrl{1}
  & \qw
  & \qw
  & \octrl{1}
  & \qw\\
\lstick{$\ket{0}_{\cal D}$}
  & \targ{}
  & \gate{Z}
  & \octrl{1}
  & \targ{}
  & \qw\\
\lstick{${\cal C}_{\mathrm{lsb}}$}
  & \qw
  & \qw
  & \gate[wires=2]{R(U)}
  & \qw
  & \qw\\
\lstick{$\mathbb C^N$}
  & \qw
  & \qw
  &
  & \qw
  & \qw
\end{quantikz}
\end{minipage}
\hfill
\begin{minipage}[t]{0.48\linewidth}
\centering
(b) Reflection $R_1$.

\vspace{0.2cm}
\begin{quantikz}[row sep=0.55cm,column sep=0.58cm]
\lstick{${\cal C}$}
  & \gate{{\sf Inc}}
  & \gate[wires=2]{R_0}
  & \gate{{\sf Inc}^\dagger}
  & \qw \\
\lstick{$\mathbb C^N$}
  & \qw
  &
  & \qw
  & \qw
\end{quantikz}
\end{minipage}
\end{minipage}
\caption{Circuit implementations of the two reflections $R_0$ and $R_1$, where $R_z(-2\sigma) = e^{i\sigma Z}$.
In (a) we decompose a binary expansion of the basis of system ${\cal C}$ in its least significant bit and the remainder, and insert an auxiliary qubit ${\cal D}$ in between. The CNOTs with open control on the infinite-dimensional ${\cal C}_{rest}$ flip the auxiliary system if the state is $\ket{0}$. These, and ${\sf Inc}$ and its adjoint in (b), are the only infinite-dimensional operations.
}
\label{fig:phase-reflection-circuits}
\end{figure}

It is clear that we can implement $R(U)$, using one controlled call to each of $U$ and $U^\dagger$, and a handful of single-qubit gates (see \Cref{fig:phase-reflection-circuits}). Then, since 
$$R_0=-\ket{0}\bra{0}_{\cal C}\otimes I+\sum_{\ell=1}^\infty \ket{\ell}\bra{\ell}_{{\cal C}_{rest}} \otimes R(U),$$
we can implement $R_0$ using an operation that checks if ${\cal C}_{rest}$ is in the state $\ket{0}$, and a call to $R(U)$, as shown in \Cref{fig:phase-reflection-circuits}, yielding:
\begin{corollary}\label{cor:R0-phase-est}
    The reflection $R_0$ can be implemented using one controlled call to each of $U$ and $U^\dagger$, two calls to $\ket{0}\bra{0}_{{\cal C}_{rest}}\otimes X+(I-\ket{0}\bra{0})_{{\cal C}_{rest}}\otimes I$, and $O(1)$ additional one- and two-qubit gates.
\end{corollary}
We can also implement $R_{1}$ similarly, by noticing that
\begin{equation}\label{eq:R_1-R_0}
    R_1 = ({\sf Inc}_{\cal C}^\dagger\otimes I) R_0 ({\sf Inc}_{\cal C}\otimes I),
\end{equation}
where ${\sf Inc}$ is as in~\Cref{eq:Inc}.
\begin{corollary}\label{cor:R1-phase-est}
    The reflection $R_1$ can be implemented using one controlled call to each of $U$ and $U^\dagger$, one call to each of ${\sf Inc}$ and ${\sf Inc}^\dagger$, two calls to $\ket{0}\bra{0}_{{\cal C}_{rest}}\otimes X+(I-\ket{0}\bra{0})_{{\cal C}_{rest}}\otimes I$, and $O(1)$ additional one- and two-qubit gates.
\end{corollary}

Circuit diagrams for both $R_0$ and $R_1$ are shown in \Cref{fig:phase-reflection-circuits}. We will later argue that the ${\sf Inc}$, ${\sf Inc}^\dagger$ and the infinite-dimensional controlled NOTs can be replaced with truncated versions, which can be efficiently implemented on a finite-dimensional space.

\paragraph{Transducer analysis.} Having defined the transducer, and described its implementation, we now analyze the induced action and transduction complexity in both the negative and positive case.

\begin{lemma}[Negative Catalyst]\label{lem:neg-catalyst-phase}
Assume $x_k\in [0,\pi]$ for all $k$. Let $S$ be as in \Cref{eq:S-decision-phase}, and for each $k\in\{0,\dots,K-1\}$, define:
    $$\ket{v_k^-}:=
    -\sum_{\ell=2}^{\infty}\left( -i\gamma_k \right)^{\ell-1}\ket{\ell}\ket{\psi_k}
    =-\sum_{\ell=2}^{\infty}\left( -i\frac{\sin\frac{\theta_k}{2}}{\cos\frac{\theta_k}{2}} \right)^{\ell-1}\ket{\ell}\ket{\psi_k}.$$
Then $\ket{1}\ket{\psi_k}+\ket{v_k^-}\overset{S}{\mapsto} -\ket{1}\ket{\psi_k}+\ket{v_k^-}$, 
    and for any $k$ such that $x_k < \tilde{s}$,
    $$\norm{\ket{v_k^-}}^2=\frac{1}{2}\left(\frac{1}{\sin(\tilde s-x_k)}-1\right).$$
\end{lemma}
This implies that whenever $x_k<\tilde s$ (i.e., in the negative case), $S:\ket{1}\ket{\psi_k}\rightsquigarrow -\ket{1}\ket{\psi_k}$. Outside of this case, the norm of $\ket{v_k^-}$ is unbounded, and so we cannot meaningfully make this claim. 
\begin{proof}
First note that
\begin{align*}
    \ket{1}\ket{\psi_k}+\ket{v_k^-} &= \ket{1}\ket{\psi_k}-\sum_{\ell=1}^{\infty}\left(-i\frac{\sin\frac{\theta_k}{2}}{\cos\frac{\theta_k}{2}}\right)^{2\ell-1}\frac{1}{\cos\frac{\theta_k}{2}}\left(\cos\frac{\theta_k}{2}\ket{2\ell}-i\sin\frac{\theta_k}{2}\ket{2\ell+1}\right)\ket{\psi_k},
\end{align*}
from which it is clear that the first term is reflected by $R_{0}$, whereas $R_0$ fixes the rest, so after applying $R_0$, we have:
\begin{align*}
    -\ket{1}\ket{\psi_k}+\ket{v_k^-} &= -\sum_{\ell=1}^{\infty}\left(-i\frac{\sin\frac{\theta_k}{2}}{\cos\frac{\theta_k}{2}}\right)^{2\ell-2}\frac{1}{\cos\frac{\theta_k}{2}}\left(\cos\frac{\theta_k}{2}\ket{2\ell-1}-i\sin\frac{\theta_k}{2}\ket{2\ell}\right)\ket{\psi_k},
\end{align*}
which is fixed by $R_{1}$. Thus:
$$\ket{1}\ket{\psi_k}+\ket{v_k^-}\overset{R_{0}}{\mapsto} -\ket{1}\ket{\psi_k}+\ket{v_k^-} 
\overset{R_{1}}{\mapsto} -\ket{1}\ket{\psi_k}+\ket{v_k^-}.$$

Before we compute $\norm{\ket{v_k^-}}^2$ in the case where $x_k<\tilde s$, note that 
\begin{align*}
    0 &\leq x_k < \tilde s\\
    \frac{\pi}{2}-\tilde s &\leq \underbrace{x_k+\frac{\pi}{2}-\tilde s}_{\theta_k} < \frac{\pi}{2}\\
    -\frac{\pi}{2} &<\theta_k<\frac{\pi}{2}
\end{align*}
Thus, $|\cos\frac{\theta_k}{2}|>| \sin \frac{\theta_k}{2}|$, which is necessary and sufficient for $\norm{\ket{v_k^-}}$ to converge. 
We can now compute:
\begin{align*}
    \norm{\ket{v_k^-}}^2
    &= \sum_{\ell=1}^{\infty} \left(\frac{\sin\frac{\theta_k}{2}}{\cos\frac{\theta_k}{2}}\right)^{2\ell}
    =\sum_{\ell=1}^\infty\left(\frac{1-\cos\theta_k}{1+\cos\theta_k}\right)^{\ell}
    = \frac{1}{1-\frac{1-\cos\theta_k}{1+\cos\theta_k}}-1
    = \frac{1}{2\cos\theta_k}-\frac{1}{2}.
\end{align*}
Plugging in $\theta_k=\frac{\pi}{2}-\tilde s + x_k$, we get $\cos\theta_k=\cos\left(\frac{\pi}{2}-(\tilde s - x_k)\right)=\sin\left(\tilde s - x_k\right)$, which gives the claimed result.
\end{proof}
We note that as $x_k$ gets closer to the bound $\tilde s$, the complexity blows up, which is also what we should expect. Fortunately, we will always be in a position where the complexity only matters for settings where there is a gap. Something similar happens in the positive case, except we reverse the roles of $\cos$ and $\sin$ in the catalyst:
\begin{lemma}[Positive Catalyst]\label{lem:pos-catalyst-phase}
Let $S$ be as in \Cref{eq:S-decision-phase}, and for each
$k\in\{0,\dots,K-1\}$, define: 
    $$\ket{v_k^+}:=
    \sum_{\ell=2}^{\infty}\left( i\frac{1}{\gamma_k} \right)^{\ell-1}\ket{\ell}\ket{\psi_k}
    =\sum_{\ell=2}^{\infty}\left( i\frac{\cos\frac{\theta_k}{2}}{\sin\frac{\theta_k}{2}} \right)^{\ell-1}\ket{\ell}\ket{\psi_k}.$$
    Then 
    $\ket{1}\ket{\psi_k}+\ket{v_k^+}\overset{S}{\mapsto} \ket{1}\ket{\psi_k}+\ket{v_k^+},$
    and for any $k$ such that $x_k>\tilde s$,
    $$\norm{\ket{v_k^+}}^2=\frac{1}{2}\left(\frac{1}{\sin(x_k-\tilde s)}-1\right).$$
\end{lemma}
Similar to the negative case, this implies that whenever $x_k>\tilde s$ (i.e. in the positive case), $S:\ket{1}\ket{\psi_k}\rightsquigarrow \ket{1}\ket{\psi_k}$. Outside of the positive case, the norm of $\ket{v_k^+}$ is unbounded. 
\begin{proof}
First note that
\begin{align*}
    \ket{1}\ket{\psi_k}+\ket{v_k^+} &= \ket{1}\ket{\psi_k}+\sum_{\ell=1}^{\infty}\left(i\frac{\cos\frac{\theta}{2}}{\sin\frac{\theta}{2}}\right)^{2\ell-1}\frac{1}{\sin\frac{\theta}{2}}\left(\sin\frac{\theta}{2}\ket{2\ell}+i\cos\frac{\theta}{2}\ket{2\ell+1}\right)\ket{\psi_k}.
\end{align*}
The first term is reflected by $R_0$, since is orthogonal to all $\cos\frac{\theta_k}{2}\ket{2\ell}-i\sin\frac{\theta_k}{2}\ket{2\ell+1}$ for $\ell\geq 1$. The second term is also orthogonal to all of these, and thus is also reflected by $R_0$, so 
$$R_0(\ket{1}\ket{\psi_k}+\ket{v_k^+})=-\ket{1}\ket{\psi_k}-\ket{v_k^+}.$$
Similarly, since:
\begin{align*}
    -\ket{1}\ket{\psi_k}-\sum_{\ell=2}^{\infty}\left(i\frac{1}{\gamma_k} \right)^{\ell-1}\ket{\ell}\ket{\psi_k} &= -\sum_{\ell=1}^{\infty}\left(i\frac{\cos\frac{\theta_k}{2}}{\sin\frac{\theta_k}{2}}\right)^{2\ell-2}\frac{1}{\sin\frac{\theta_k}{2}}\left(\sin\frac{\theta_k}{2}\ket{2\ell-1}+i\cos\frac{\theta_k}{2}\ket{2\ell}\right)\ket{\psi_k}.
\end{align*}
It is clearly orthogonal to all states $\cos\frac{\theta}{2}\ket{2\ell-1}-i\sin\frac{\theta}{2}\ket{2\ell}$, so $R_{1}$ reflects the state, resulting in: 
$$R_1R_0(\ket{1}\ket{\psi_k}+\ket{v_k^+})=\ket{1}\ket{\psi_k}+\ket{v_k^+}.$$
as claimed. 

Before computing $\norm{\ket{v_k^+}}^2$ in the case where $x_k>\tilde s$, note that:
\begin{align*}
    \tilde s &< x_k \leq \pi\\ \tilde s &< x_k < \pi+\tilde s\\
    \frac{\pi}{2} &< x_k+\frac{\pi}{2}-\tilde s < \frac{3\pi}{2} 
\end{align*}
so $|\cos\frac{\theta_k}{2}|<|\sin\frac{\theta_k}{2}|$, which is necessary and sufficient for $\norm{\ket{v_k}}$ to converge. We can now compute:
$$\norm{\ket{v_k^+}}^2
=\sum_{\ell=1}^\infty\left(\frac{\cos\frac{\theta_k}{2}}{\sin\frac{\theta_k}{2}}\right)^{2\ell}
=\sum_{\ell=1}^{\infty}\left(\frac{1+\cos\theta_k}{1-\cos\theta_k}\right)^\ell=\frac{1}{1-\frac{1+\cos\theta_k}{1-\cos\theta_k}}-1=\frac{1}{-2\cos\theta_k}-\frac{1}{2}$$
Plugging in $\theta_k=\frac{\pi}{2}-\tilde s + x_k$, we get $-\cos\theta_k=\cos\left(\frac{\pi}{2}-(x_k-\tilde s)\right)=\sin(x_k-\tilde s)$, which gives the claimed result.
\end{proof}

\subsection{Direct construction of composed transducer}\label{sec:direct-construction}

In this section, we prove \Cref{thm:main-transducer} by exhibiting a transducer for threshold max eigenphase. 
It combines the simple transducer for decision amplitude amplification (\Cref{def:decision-aa}) from \Cref{sec:decision-aa-transducer} with the infinite transducer for decision phase estimation (\Cref{def:decision-phase-est}) from the previous section. We give the combined transducer explicitly. It is based on the infinite graph in \Cref{fig:composed-transducer}.
\begin{figure}
    \centering
\begin{tikzpicture}
    \draw (-4,0)--(10,0);
    \filldraw (-2,0) circle (.05);
    \filldraw (0,0) circle (.05);
    \draw (0,0)--(1.5,1.5);
    \draw (0,0)--(1.5,-1.5);
    \filldraw (2.5,0) circle (.05);
    \filldraw (5,0) circle (.05);
    \filldraw (7,0) circle (.05);
    \filldraw (9.5,0) circle (.05);

    \node at (-3,.25) {$\sqrt{\gamma}\ket{\bar{1}}\ket{0}$};
    \node at (-1,.25) {$\sqrt{\gamma}\ket{0}\ket{0}$};
    \node[rotate=45] at (.75,1.1) {$\alpha_0\ket{1}\ket{\psi_0}$};
    \node[rotate=45] at (1.75,1.75) {$\dots$};
    \node[rotate=-45] at (.75,-1.1) {$\alpha_{K-1}\ket{1}\ket{\psi_{K-1}}$};
    \node[rotate=-45] at (1.75,-1.75) {$\dots$};

    \node[rotate=-60] at (1.25,.75) {$\dots$};
    \node[rotate=60] at (1.25,-.75) {$\dots$};

    \node at (1.3,.25) {$\alpha_k\ket{1}\ket{\psi_k}$};
    \node at (3.75,.25) {$\alpha_k\gamma_k\ket{2}\ket{\psi_k}$};

    \node[fill=white] at (6,0) {$\dots$};
    \node at (8.25,.25) {$\alpha_k\gamma_k^{\ell-1}\ket{\ell}\ket{\psi_k}$};
    \node at (10.5,0) {$\dots$};
\end{tikzpicture}
\caption{Visualization of the transducer for threshold max eigenphase. It can be seen as a composition of the transducers visualized in \Cref{fig:aa-transducer} and \Cref{fig:decision-phase-graph}, with an extra edge $\ket{\bar{1},0}$ added, whose purpose is to make the transducer canonical, by ensuring that the boundary is orthogonal to the space acted on by the oracle.}
    \label{fig:composed-transducer}
\end{figure}
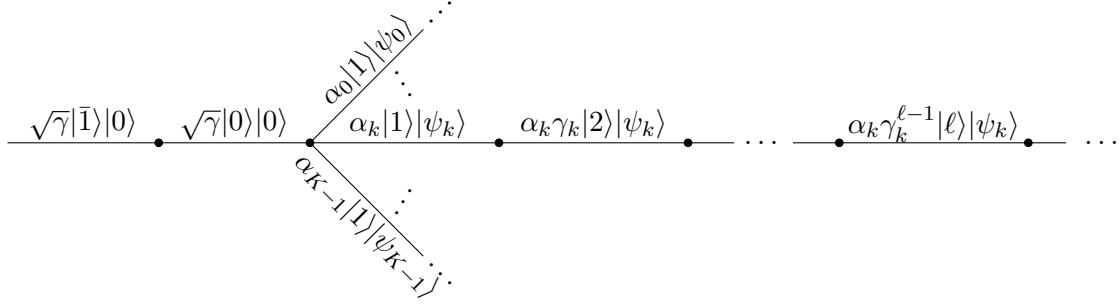
Setting $\tilde s = s-\delta/2$, the following is just as in \Cref{sec:decision-phase-est-transducer}:
$$\gamma_k=\frac{\sin\frac{\theta_k}{2}}{\cos\frac{\theta_k}{2}}
\mbox{ where }\theta_k=\frac{\pi}{2}-s+\delta/2+x_k.$$
The transducer acts on 
$$\mathrm{span}\{\ket{\bar{1}}_{\overline{\cal C}}\ket{0}\}\oplus \mathrm{span}\{\ket{\ell}_{\overline{\cal C}}:\ell\in \mathbb{N}_0\}\otimes \mathbb{C}^N,$$
with $\overline{\cal C}=\mathrm{span}\{\ket{\bar{1}}\}\oplus {\cal C}$.
The edge $\ket{\bar{1},0}$ is the boundary. We could have gone without this extra edge and just used $\ket{0}\ket{0}$ as the boundary, but then this would not have been canonical, because the oracle $O_\gamma(A)$ acts non-trivially on this space.

Define 
$$\ket{\varphi}=\frac{1}{\sqrt{2}}(\ket{\bar{1},0}+\ket{0,0}),$$
which is the state we associate with the first vertex.
Note that 
$$\ket{\varphi_\gamma(A)}=\sqrt{\frac{\gamma}{1+\gamma}}\ket{0,0}+\frac{1}{\sqrt{1+\gamma}}\ket{1}\ket{\psi}
=\sqrt{\frac{\gamma}{1+\gamma}}\ket{0,0}+\frac{1}{\sqrt{1+\gamma}}\sum_{k=0}^{K-1}\alpha_k\ket{1}\ket{\psi_k}$$
is the state we associate with the second vertex in \Cref{fig:composed-transducer}.
The transducer will be the product of reflections $S=R_1'R_0'$, where: 
\begin{align*}
    R_1' &=2\ket{\varphi}\bra{\varphi}+2\sum_{k=0}^{K-1}\sum_{\ell=1}^\infty\left(\cos\frac{\theta_k}{2}\ket{2\ell-1}-i\sin\frac{\theta_k}{2}\ket{2\ell}\right)\left(\cos\frac{\theta_k}{2}\bra{2\ell-1}+i\sin\frac{\theta_k}{2}\bra{2\ell}\right)\otimes\Pi_k - I\\
    R_0' &= 2\ket{\varphi_\gamma(A)}\bra{\varphi_\gamma(A)}\\
    &\qquad +2\sum_{k=0}^{K-1}\sum_{\ell=1}^\infty\left(\cos\frac{\theta_k}{2}\ket{2\ell}-i\sin\frac{\theta_k}{2}\ket{2\ell+1}\right)\left(\cos\frac{\theta_k}{2}\bra{2\ell}+i\sin\frac{\theta_k}{2}\bra{2\ell+1}\right)\otimes\Pi_k - I
\end{align*}
which are reflections around odd and even vertices respectively.  Note that $R_1'$ is related to $R_1$ from \Cref{sec:decision-phase-est-transducer} by:
\begin{equation}
    R_1'=(I-2\ket{\varphi}\bra{\varphi})R_1
\end{equation}
and $R_0'$ is related to the $R_0$ from that section by:
\begin{equation}
    R_0'=R_0{(I-2\ket{\varphi_\gamma(A)}\bra{\varphi_\gamma(A)})} = R_0 O_\gamma(A).\label{eq:R0-prime}
\end{equation}
Then, as depicted in \Cref{fig:canonical-transducer-circuit},
\begin{equation}\label{eq:S-circ-full}
    S(O_\gamma(A))= \underbrace{(I-2\ket{\varphi}\bra{\varphi}) R_1 R_0}_{=:S^\circ} O_\gamma(A).
\end{equation} 
To conclude that this is a canonical form transducer with respect to the oracle $O_\gamma(A)$, we need to exhibit ${\cal H}$, ${\cal L}^\bullet$ and ${\cal L}^\circ$. In particular, the space ${\cal L}^\bullet$ on which the oracle acts non-trivially should be orthogonal to ${\cal H}$, which will be the boundary edge -- that is the reason for the introduction of the extra edge $\ket{\bar{1},0}$. We thus define:
\begin{align*}
    {\cal H}&=\mathrm{span}\{\ket{\bar{1},0}\}\\
    {\cal L}^{\bullet} &=\mathrm{span}\{\ket{0},\ket{1}\}\otimes \mathbb{C}^N\\
    {\cal L}^{\circ} &=\mathrm{span}\{\ket{\ell}:\ell\geq 2\}\otimes \mathbb{C}^N.
\end{align*}

\paragraph{Implementation of truncated transducer.} We first formally observe that the oracle $O_\gamma(A)$ has an efficient implementation. This is not, strictly speaking, needed for the proof of \Cref{thm:main-transducer}, but we will use it later, when we make use of \Cref{thm:main-transducer} in \Cref{sec:algorithm}. 

\begin{lemma}\label{lem:A-impl}
    $O_\gamma(A)$ can be implemented using one controlled call to each of $A$ and $A^\dagger$, and $O(\log N)$ additional one- and two-qubit gates.
\end{lemma}
\begin{proof}
We can implement a map $G$ such that $G\ket{0}=\ket{\varphi_\gamma(A)}$ using a single qubit gate that acts as:
$$\ket{0}\mapsto \sqrt{\frac{\gamma}{1+\gamma}}\ket{0}+\frac{1}{\sqrt{1+\gamma}}\ket{1},$$
a controlled call to $A$ to map $\ket{0}\mapsto \ket{\psi}$, and a map that generates $\ket{\bar{1},0}$, which we will assume without loss of generality is trivial. Then note that:
$$I - 2\ket{\varphi_\gamma(A)}\bra{\varphi_\gamma(A)} = G(I - 2\ket{0^{\log N + 1}}\bra{0^{\log N + 1}})G^\dagger,$$
and $I - 2\ket{0^{\log N + 1}}\bra{0^{\log N + 1}}$ can be implemented using $O(\log N)$ gates. 
\end{proof}

\begin{figure}[t]
\centering
\resizebox{0.98\textwidth}{!}{%
\begin{quantikz}[row sep=0.65cm,column sep=0.55cm]
\lstick{$\overline{\cal C}=\mathrm{span}\{\ket{\bar 1},\ket0,\ket1,\ldots\}$}
  & \gate[wires=2]{O_\gamma(A)}
  & \gate[wires=2]{R_0}
    \gategroup[wires=2,steps=5,
      style={dashed,rounded corners,inner xsep=4pt,inner ysep=4pt},
      background,
      label style={label position=below,anchor=north,yshift=-0.22cm}]
      {$S^\circ$}
  & \gate{{\sf Inc}}
  & \gate[wires=2]{R_0}
  & \gate{{\sf Inc}^\dagger}
  & \gate[wires=2]{I-2\ket{\varphi}\bra{\varphi}}
  & \qw \\
\lstick{$\mathbb C^N$}
  &
  &
  & \qw
  &
  & \qw
  &
  & \qw
\end{quantikz}%
}
\caption{Circuit diagram for the canonical transducer
$S(O_\gamma(A))=(I-2\ket{\varphi}\bra{\varphi})R_1R_0O_\gamma(A)$.}
\label{fig:canonical-transducer-circuit}
\end{figure}
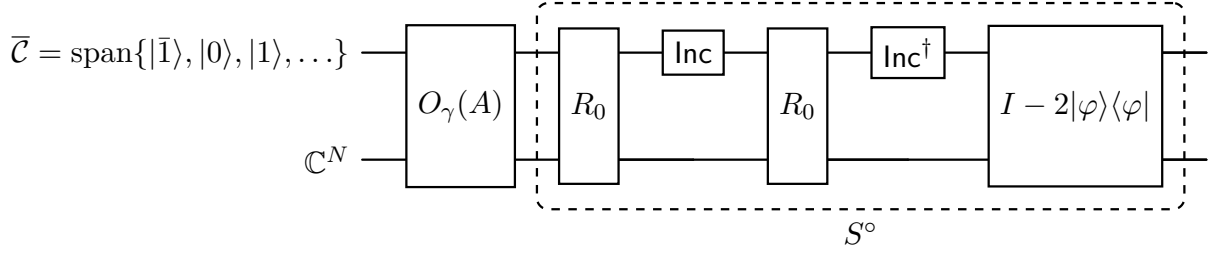

We have seen, in \Cref{eq:S-circ-full} that $S^\circ=(I-2\ket{\varphi}\bra{\varphi})R_1R_0$, so its implementation follows from the implementations of $R_1$ and $R_0$ in \Cref{cor:R1-phase-est} and \Cref{cor:R0-phase-est} (see \Cref{fig:phase-reflection-circuits}), but these use the infinite-dimensional operations ${\sf Inc}$ and $CNOT$ with infinite-dimensional control. In the following, we show that we can truncate these in the natural way, to get a truncation that is exact, in the sense of \Cref{def:perfect-truncation}. 

\begin{lemma}\label{lem:finite-work-unitary} For any even $D\geq 2$:
the oracle $O_\gamma(A)$ is an exact $(D,1)$-truncation of itself; and
    there is an exact $(D,3)$-truncation of $S^\circ$ that can be implemented using $O(1)$ controlled calls to $U$ and $U^\dagger$, and $O(\log D)$ additional one- and two-qubit gates. 
\end{lemma}

\begin{proof}
First of all, the oracle acts nontrivially only on the states $\ket{0},\ket{1}$ of the register ${\cal C}$, so it can increase $\ell$ by at most 1, satisfying the first part of \Cref{def:perfect-truncation} for $m=1$ (and any $D$). The second part is trivially satisfied since we made no change to the oracle. 

Next, we truncate $S^\circ = R_1'R_0$ to  $S^{\circ [D]}$ for some even $D$, by replacing:
\begin{enumerate}
    \item ${\cal C}$ with ${\cal C}^{[D]}:=\mathrm{span}\{0,\dots,D-1\}$
    \item ${\cal C}_{rest}$ with ${\cal C}_{rest}^{[D]}:=\mathrm{span}\{0,\dots,D/2-1\}$
    \item ${\sf Inc}$ with ${\sf Inc}^{[D]}$ (see \Cref{eq:Inc-D})
    \item the infinite-dimensional CNOT, $\ket{0}\bra{0}_{{\cal C}_{rest}}\otimes X_{\cal D} + (I-\ket{0}\bra{0})_{{\cal C}_{rest}}\otimes I_{\cal D}$, with the same operation, ${\sf CNOT}^{[D]}$ on ${\cal C}_{rest}^{[D]}\otimes {\cal D}$.
\end{enumerate}
It is simple to observe that both ${\sf Inc}^{[D]}$ and ${\sf CNOT}^{[D]}=\sum_{\ell=1}^{D/2-1}\ket{\ell}\bra{\ell}\otimes I_{\cal D}+\ket{0}\bra{0}\otimes X_{\cal D}$ can be implemented in $O(\log D)$ one- and two-qubit gates. Thus, by \Cref{cor:R1-phase-est} and \Cref{cor:R0-phase-est}, $S^{\circ [D]}$ can be implemented using one call to each of $U$ and $U^\dagger$, and $O(\log D)$ additional gates. 

It is easy to verify from the definition that $R_0$ and $R_1'$ both map ${\cal H} \op {\cal L}_{\leq m}$ to ${\cal H} \op {\cal L}_{\leq m+1}$, so their composition maps ${\cal H} \op {\cal L}_{\leq m}$ to ${\cal H} \op {\cal L}_{\leq m+2}$, satisfying the first part of \Cref{def:perfect-truncation}.  

For the second part, consider truncating $S^{\circ}=(I-2\ket{\varphi}\bra{\varphi})R_1 R_0$, as described above. 
The operator $R_0=-\ket{0}\bra{0}\otimes I + \sum_{\ell=1}^\infty\ket{\ell}\bra{\ell}\otimes R(U)-\ket{\bar{1},0}\bra{\bar{1},0}$ has trivial truncation $-\ket{0}\bra{0}\otimes I + \sum_{\ell=1}^{D/2-1}\ket{\ell}\bra{\ell}\otimes R(U) -\ket{\bar{1},0}\bra{\bar{1},0}$, which clearly acts identically on ${\cal H}\oplus {\cal L}_{\leq D-1}$. $R_1$, as given in \cref{eq:R_1-R_0} uses ${\sf Inc}$ and ${\sf Inc}^\dagger$ once. It is easy to see that replacing these by ${\sf Inc^{[D]}}$ and ${\sf Inc^{[D],\dagger}}$ we get an operator $R_1^{[D]}$ that acts as $R_1$ on ${\cal L}_{\leq D-2}$. 
Finally, $(I-2\ketbra{\varphi}{\varphi})$ by definition only acts nontrivially below level~$1$.
\end{proof}

We have established \Cref{item:spaces} of \Cref{thm:main-transducer} by definition, and \Cref{item:more-spaces} was established in the previous lemma.

\paragraph{Analysis of the transducer.} We now analyze the complexities of the transducer, proving \Cref{item:pos}, \Cref{item:neg}, and \Cref{item:complexities} of \Cref{thm:main-transducer}, in \Cref{lem:neg-catalyst} and \Cref{lem:pos-catalyst} below.

\begin{lemma}[Negative Catalyst]\label{lem:neg-catalyst}
Suppose $\norm{\Pi_{>s-\delta}\ket{\psi}}=0$, which is the negative case.
Then there exists a catalyst $\ket{v}\in {\cal L}$ such that $S(\ket{\bar{1},0}+\ket{v})=-\ket{\bar{1},0}+\ket{v}$, $\norm{\ket{v}}^2\leq O(\frac{1}{\gamma\delta})$, and $\norm{\Pi_{{\cal L}^\bullet}\ket{v}}^2\leq 1+\frac{1}{\gamma}$.
\end{lemma}
\begin{proof}
    Let 
    $$\ket{v}=-\ket{0,0}+\frac{1}{\sqrt{\gamma}}\sum_{k=0}^{K-1}\alpha_k\ket{{\bar{v}_k^-}}
    \mbox{ where }
    \ket{{\bar{v}_k^-}}=-\sum_{\ell=1}^\infty(-i\gamma_k)^{\ell-1}\ket{\ell}\ket{\psi_k}.$$
Note that
\begin{align*}
    \ket{\bar{1},0}+\ket{v} &= \ket{\bar{1},0}-\ket{0,0}-\frac{1}{\sqrt{\gamma}}\sum_{k=0}^{K-1}\alpha_k\left(\ket{1}\ket{\psi_k}-\ket{v_k^-}\right),
\end{align*}
where $\ket{v_k^-}$ is as in \Cref{lem:neg-catalyst-phase}. Applying $O_\gamma(A)$ reflects $\ket{\varphi_\gamma(A)}\propto\ket{0,0}+\frac{1}{\sqrt{\gamma}}\sum_k\alpha_k\ket{1}\ket{\psi_k}$, and fixes everything orthogonal, giving:
\begin{equation*}\label{eq:O-gamma-A-application}
    O_\gamma(A)(\ket{\bar{1},0}+\ket{v}) = \ket{\bar{1},0}+\ket{0,0}+\frac{1}{\sqrt{\gamma}}\sum_{k=0}^{K-1}\alpha_k\left(\ket{1}\ket{\psi_k}+\ket{v_k^-}\right).
\end{equation*}
By \Cref{lem:neg-catalyst-phase}, $R_1R_0(\ket{1}\ket{\psi_k}+\ket{v_k^-})=-\ket{1}\ket{\psi_k}+\ket{v_k^-}$. It is clear from the definitions of $R_1$ and $R_0$ that each of them reflect $\ket{\bar{1},0}$ and $\ket{0,0}$, meaning their product acts as the identity, giving: 
\begin{equation*}
    R_1R_0O_\gamma(A)(\ket{\bar{1},0}+\ket{v}) = \ket{\bar{1},0}+\ket{0,0}-\frac{1}{\sqrt{\gamma}}\sum_{k=0}^{K-1}\alpha_k\left(\ket{1}\ket{\psi_k}-\ket{v_k^-}\right).
\end{equation*}
Finally, applying $(I-2\ket{\varphi}\bra{\varphi})$ reflects $\ket{\varphi}\propto \ket{\bar{1},0}+\ket{0,0}$ and fixes everything orthogonal, giving:
\begin{equation*}
    (I-2\ket{\varphi}\bra{\varphi})R_1R_0O_\gamma(A)(\ket{\bar{1},0}+\ket{v}) = -\ket{\bar{1},0}-\ket{0,0}-\frac{1}{\sqrt{\gamma}}\sum_{k=0}^{K-1}\alpha_k\left(\ket{1}\ket{\psi_k}-\ket{v_k^-}\right)=-\ket{\bar{1},0}+\ket{v}.
\end{equation*}

 To upper bound the complexity, we first have:
\begin{equation}
\begin{split}
    \norm{\ket{v}}^2 &= 1+\frac{1}{\gamma}\norm{\sum_{k}\alpha_k(\ket{1}\ket{\psi_k}+\ket{v_k^-}}^2
    = 1+\frac{1}{\gamma}\left(\sum_k|\alpha_k|^2+\sum_k|\alpha_k|^2\norm{\ket{v_k^-}}^2\right)\\
    &\leq 1+\frac{1}{\gamma}\left(1+\sum_k|\alpha_k|^2\frac{1}{2}\left(\frac{1}{\sin(\tilde s - x_k)}-1\right)\right)
\end{split}\label{eq:norm-v-2}
\end{equation}
since the states $\ket{1}\ket{\psi_k}+\ket{v_k^-}$ are pairwise orthogonal, as are the two terms $\ket{1}\ket{\psi_k}$ and $\ket{v_k^-}$, and using the upper bound on $\norm{\ket{v_k^-}}^2$ from \Cref{lem:neg-catalyst-phase}, which we can apply because $x_k\leq s-\delta$ for all $k$ in the negative case.

We can compute
\begin{align*}
    \frac{1}{\sin(\tilde s - x_k)} = \frac{1}{\sin(s-\delta/2 - x_k)}\leq \frac{\pi}{2(s-\delta/2-x_k)},
\end{align*}
as long as $s-\delta/2-x_k\in [0,\pi/2]$. Indeed, since $\Pi_{>s-\delta}\ket{\psi}=0$ by assumption, we have that for any $k$ such that $\alpha_k\neq 0$, $x_k\leq s-\delta < s-\delta/2$, so $s-\delta/2-x_k\geq 0$, and $s-\delta/2-x_k\leq \frac{\pi}{2}$ follows from $s\leq \pi/2$, $x_k\geq 0$ (by the windowing promise) and $\delta\geq 0$. Thus, continuing from \Cref{eq:norm-v-2}:
\begin{align*}
    \norm{\ket{v}}^2 
    &\leq 1+\frac{1}{\gamma}\left(1+\sum_k|\alpha_k|^2\frac{1}{2}\left(\frac{\pi}{2(s-\delta/2 - x_k)}-1\right)\right)\\
    &\leq 1+\frac{1}{\gamma}\left(1+\frac{1}{2}\left(\frac{\pi}{2(s-\delta/2 - (s-\delta))}-1\right)\right)\\
    &= 1+\frac{1}{2\gamma}\left(1+\frac{\pi}{\delta}\right) = O\left(\frac{1}{\gamma\delta}\right),
\end{align*}
where we once again used $x_k\leq s-\delta$ for any $k$ such that $\alpha_k\neq 0$, since we're in the negative case.

To upper bound the Las Vegas query complexity, we have:
\begin{equation*}
    \norm{\Pi_{{\cal L}^\bullet}\ket{v}}^2 = \norm{\ket{0,0}+\frac{1}{\sqrt\gamma}\sum_k\alpha_k\ket{1}\ket{\psi_k}}^2 = 1+\frac{1}{\gamma}. \qedhere
\end{equation*}
\end{proof}

\begin{lemma}[Positive Catalyst]\label{lem:pos-catalyst}
Suppose $\norm{\Pi_{>s}\ket{\psi}}\geq \gamma$, which is the positive case.
    Then there exists a catalyst $\ket{v}\in {\cal L}$ such that $S(\ket{\bar{1},0}+\ket{v})=\ket{\bar{1},0}+\ket{v}$, $\norm{\ket{v}}^2\leq O(\frac{1}{\gamma\delta})$, and $\norm{\Pi_{{\cal L}^\bullet}\ket{v}}^2\leq 1+\frac{1}{\gamma}$.
\end{lemma}
\begin{proof}
    Let 
    $$\ket{v}=-\ket{0,0}+{\sqrt{\gamma}}\sum_{k:x_k\geq s}\frac{\alpha_k}{{\eps}}\ket{\bar{v}_k^+}
    \mbox{ where }
    \ket{\bar{v}_k^+}=\sum_{\ell=1}^\infty\left(i\frac{1}{\gamma_k}\right)^{\ell-1}\ket{\ell}\ket{\psi_k},$$
    and where $\eps = \sum_{k:x_k\geq s}|\alpha_k|^2=\norm{\Pi_{>s}\ket{\psi}}^2\geq \gamma^2$.
Note that
\begin{equation}\label{eq:orig}
    \ket{\bar{1},0}+\ket{v}=\ket{\bar{1},0}-\ket{0,0}+{\sqrt{\gamma}}\sum_{k:x_k\geq s}\frac{\alpha_k}{{\eps}}(\ket{1}\ket{\psi_k}+\ket{v_k^+})
\end{equation}
where $\ket{v_k^+}$ is as in \Cref{lem:pos-catalyst-phase}.

First, $O_\gamma(A)$ fixes anything orthogonal to $\ket{\varphi_\gamma(A)}\propto \sqrt{\gamma}\ket{0,0}+\sum_k\alpha_k\ket{1}\ket{\psi_k}$, which certainly includes $\ket{\bar{1},0}$ and $\ket{v_k^+}$ for all $k$ (which are only supported on $\ell\geq 2$ in register $\cal C$). We also have:
\begin{align*}
    \bra{\varphi_\gamma(A)}\left(\ket{0,0}-{\sqrt{\gamma}}\sum_{k:x_k\geq s}\frac{\alpha_k}{{\eps}}\ket{1}\ket{\psi_k}\right) &= \sqrt{\gamma} - \sqrt{\gamma}\sum_{k:x_k\geq s}\frac{|\alpha_k|^2}{\eps} = \sqrt{\gamma}-\sqrt{\gamma}=0,
\end{align*}
by definition of $\eps$. We have thus established that $O_\gamma(A)$ leaves $\ket{\bar{1},0}+\ket{v}$ unchanged. 

By \Cref{lem:pos-catalyst-phase}, $R_1R_0(\ket{1}\ket{\psi_k}+\ket{v_k^+})=\ket{1}\ket{\psi_k}+\ket{v_k^+}$, and it is clear from the definitions of $R_0$ and $R_1$ that each of them reflects each of $\ket{0,0}$ and $\ket{\bar{1},0}$, so the product keeps both unchanged. Thus $R_1R_0$ also leaves $\ket{\bar{1},0}+\ket{v}$ unchanged. 

Finally, as $I-2\ket{\varphi}\bra{\varphi}$ leaves everything orthogonal to $\ket{\varphi}\propto \ket{\bar{1},0}+\ket{0,0}$ fixed, it too leaves $\ket{\bar{1},0}+\ket{v}$, so we can conclude that $S=(I-2\ket{\varphi}\bra{\varphi})R_1R_0 O_\gamma(A)$ leaves $\ket{\bar{1},0}+\ket{v}$ fixed, as claimed.

To upper bound the complexity, we first have:
\begin{equation}\label{eq:norm-v-2-pos}
\begin{split}
    \norm{\ket{v}}^2 &= 1+\frac{\gamma}{\eps^2}\norm{\sum_{k:x_k\geq s}\alpha_k(\ket{1}\ket{\psi_k}+\ket{v_k^+})}^2
    = 1+\frac{\gamma}{\eps^2}\left(\sum_{k:x_k\geq s}|\alpha_k|^2+\sum_{k:x_k\geq s}|\alpha_k|^2\norm{\ket{v_k^+}}^2\right)\\
    &\leq 1+\frac{\gamma}{\eps^2}\left(\eps+\sum_{k:x_k\geq s}|\alpha_k|^2\frac{1}{2}\left(\frac{1}{\sin(x_k-\tilde s)}-1\right)\right)
\end{split}
\end{equation}
since the states $\ket{1}\ket{\psi_k}+\ket{v_k^+}$ are pairwise orthogonal, as are the two terms $\ket{1}\ket{\psi_k}$ and $\ket{v_k^+}$, and using the upper bound on $\norm{\ket{v_k^+}}^2$ from \Cref{lem:pos-catalyst-phase}, which we can apply because $x_k\geq s$ for all $k$ in the sum.

We can compute
\begin{align*}
    \frac{1}{\sin(x_k-\tilde s)} = \frac{1}{\sin(x_k-s+\delta/2)}\leq \frac{\pi}{2(x_k-s+\delta/2)}
    \leq \frac{\pi}{2(s-s+\delta/2)} = \frac{\pi}{\delta},
\end{align*}
as long as $x_k-s+\delta/2\in [0,\pi/2]$. Indeed, since $x_k\geq s > s-\delta/2$ for any $k$ in the sum, $x_k-s+\delta/2>0$, and $x_k-s+\delta/2\leq \frac{\pi}{2}$ follows from $x_k\leq \pi/2$ and $s\geq \delta$. Thus, continuing from \Cref{eq:norm-v-2-pos}:
\begin{align*}
    \norm{\ket{v}}^2 
    &\leq 1+\frac{\gamma}{\eps^2}\left(\eps+\sum_{k:x_k\geq s}|\alpha_k|^2\frac{1}{2}\left(\frac{\pi}{\delta}-1\right)\right)
    =1+\frac{\gamma}{\eps^2}\left(\eps+\eps\frac{1}{2}\left(\frac{\pi}{\delta}-1\right)\right)
    =O\left(\frac{\gamma}{\eps\delta}\right).
\end{align*}
Since $\eps\geq \gamma^2$ in the positive case, this is $O(\frac{1}{\gamma\delta})$, as claimed.

Finally, to upper bound the Las Vegas complexity, we have:
\begin{equation*}
    \norm{\Pi_{{\cal L}^\bullet}\ket{v}}^2 = \norm{-\ket{0,0}+\frac{\sqrt\gamma}{\eps}\sum_{k:x_k\geq s}\alpha_k\ket{1}\ket{\psi_k}}^2 = 1+\frac{\gamma}{\eps^2}\eps=1+\frac{\gamma}{\eps}.
\end{equation*}
Since $\eps\geq \gamma^2$, this is at most $1+\frac{1}{\gamma}$.
\end{proof}

\begin{remark}
    The transducer we designed in this section is the walk operator of an infinite-dimensional quantum walk, and although we have analyzed its complexity as a transducer, if we had instead analyzed it as a standard quantum walk, like that in~\cite{belovs2013ElectricWalks}, we would likely find that it approximately solves the threshold max eigenphase problem using a number of calls to the walk operator that scales like $O(W(S))$. However, implementing the full walk operator requires calling $A$ and its adjoint, and so the number of queries to $A$ would scale like $O(W(S))$ as well. 
\end{remark}

\section{Quantum algorithm for max eigenphase estimation}\label{sec:algorithm}

The following corollary of \Cref{thm:main-transducer} uses the transducer from the previous section to get a quantum algorithm for the threshold max eigenphase problem (\Cref{def:threshold}), which is the main subroutine in \Cref{alg:max-eigenphase-search} below for max eigenphase estimation, which we analyze in \Cref{thm:main}, thus proving \Cref{thm:main-intro}. 

\begin{corollary}\label{cor:threshold-eigenphase-alg}
    For any $\epsilon\in (0,1/2)$, $s\in (0,\pi/2]$, $\delta'\in (0,s)$, and $\gamma\in (0,1]$, there exists an $\epsilon$-error quantum algorithm, ${\cal D}_{\epsilon,s,\delta',\gamma}(A,U)$, for the threshold max eigenphase problem, that makes $O(\frac{1}{\gamma}\log\frac{1}{\epsilon})$ calls to $A$, $O(\frac{1}{\gamma\delta'}\log\frac{1}{\epsilon})$ calls to $U$, and $O((\frac{1}{\gamma\delta'}\log\frac{1}{\gamma\delta'}+\frac{1}{\gamma}\log N)\log\frac{1}{\epsilon})$ additional one- and two-qubit gates.
\end{corollary}

\begin{proof}
By \Cref{thm:transducer-implementation}, we can turn the transducer in \Cref{thm:main-transducer} into a bounded-error quantum algorithm  for the threshold eigenphase problem, that makes $K_O=O(1/\gamma)$ calls to $O_{\gamma}(A)$ and $K_S=O(1/(\gamma\delta'))$ calls to $S^\circ$ and additional gates.
By \Cref{lem:finite-work-unitary}, we can replace each call to $S^\circ$ with an exact $(D,3)$-truncation, $S^{\circ [D]}$, and, since $O$ is an exact $(D,1)$-truncation of itself, as long as $D>2K_S+K_O = \bigO{\frac{1}{\gamma\delta'}}$, the algorithm's behaviour will not change, and thus, it will still solve the threshold eigenphase problem with bounded error. Use the $(D,3)$-truncation referred to in \Cref{item:more-spaces} of \Cref{thm:main-transducer}, and choose a sufficiently large $D=\Theta(\frac{1}{\gamma\delta'})$. Then since each of the $O(1/\gamma)$ calls to $O_\gamma(A)$ can be implemented using $O(1)$ controlled calls to $A$ and its adjoint, and $O(\log N)$ additional one- and two-qubit gates (\Cref{lem:A-impl}); and each of the $O(1/(\gamma\delta'))$ calls to $S^{\circ [D]}$ can be implemented using $O(1)$ controlled calls to $U$ and its adjoint and $O(\log D)=O(\log\frac{1}{\gamma\delta'})$ additional one- and two-qubit gates (\Cref{item:more-spaces}), the total resources needed for this bounded-error quantum algorithm are:
\begin{itemize}
    \item $O(\frac{1}{\gamma})$ controlled calls to $A$ and $A^\dagger$,
    \item $O(\frac{1}{\gamma\delta'})$ controlled calls to $U$ and $U^\dagger$,
    \item $O(\frac{1}{\gamma}\log N+\frac{1}{\gamma\delta'}\log \frac{1}{\gamma\delta'})$ additional gates.
\end{itemize}
Repeating $O(\log\frac{1}{\epsilon})$ times and taking the majority yields the desired algorithm ${\cal D}_{\epsilon,s,\delta',\gamma}(A,U)$. 
\end{proof}

We now describe an algorithm for max eigenphase estimation. It uses the threshold eigenphase
algorithm from \Cref{cor:threshold-eigenphase-alg} as a subroutine in an interval search. Apart from the threshold eigenphase subroutine, the algorithm is classical. At the start of round $r$, the algorithm has an interval $[\ell_r,u_r]$ of possible values for the
max eigenphase. Let $L_r=u_r-\ell_r$ be its length. The algorithm then uses the threshold subroutine at
\begin{align*}
    s_r=\ell_r+\frac{2L_r}{3},
    \qquad
    \delta'_r=\frac{L_r}{3}.
\end{align*}
Note that the threshold problem has a region $[s_r-\delta'_r,s_r]$ (in this case the middle third of the interval) where the algorithm has no guarantees: when the max eigenphase lies in this region, neither promise in \Cref{def:threshold} need hold, and either
outcome may occur. Therefore, outcome $1$ gives the updated interval
$[s_r-\delta'_r,u_r]$, whereas outcome $0$ gives $[\ell_r,s_r]$. These updated intervals have lengths $2L_r/3$, so each round shrinks the interval by a factor $2/3$, and the total number of rounds is $R=\ceil{\log_{3/2}(\pi/(4\delta))}=\bigO{\log\frac{1}{\delta}}$.

We say we have an error in round $r$ if the positive or negative case holds for the threshold eigenphase problem, but the outcome is incorrect (note that if neither promise holds, it does not matter what the outcome is, and the updated intervall will contain the max eigenphase). It is then clear that if no round has an error, the algorithm returns an accurate estimate, since the updated interval contains the maximal eigenphase in each round.
It remains to choose the per-round error budgets $\eps_r$, which must satisfy $\sum_{r=1}^R\eps_r\leq \frac{1}{3}$ so that, by a union bound, with constant probability no round errs. Since a call to ${\cal D}_{\eps,s,\delta',\gamma}$ costs a factor $\log\frac{1}{\eps}$ in all resources, and since the calls to $A$ -- unlike the calls to $U$ -- cost the same $\Theta(\frac{1}{\gamma}\log\frac{1}{\eps})$ in \emph{every} round, the error budget should be spread as evenly as possible over the rounds, subject to summability. We therefore set
\begin{align*}
    \eps_r=\frac{1}{5k_r^2},\qquad\mbox{where } k_r \defeq R-r+1
\end{align*}
is the number of rounds remaining after (and including) round $r$. In this way the final rounds -- which dominate the cost in calls to $U$ -- run at constant error, while $\sum_r\eps_r\leq\frac{\pi^2}{30}<\frac{1}{3}$, since $\sum_{r=1}^{\infty}\frac{1}{r^2}=\pi^2/6$.
The algorithm is stated in \Cref{alg:max-eigenphase-search}, and one interval update is illustrated in \Cref{fig:interval-update}.

\begin{algorithm}[H]
\caption{Interval search for maximum-eigenphase estimation.}
\label{alg:max-eigenphase-search}
\begin{algorithmic}[1]
    \Require $U,A,\gamma,\delta$ satisfying the assumptions of \Cref{def:estimation}
    \Ensure An estimate $\hat\theta$ of $\theta_{\max}$
    \State $(\ell_1,u_1) \gets (0,\pi/2)$
    \State $R \gets \ceil{\log_{3/2}(\pi/(4\delta))}$
    \For{$r=1,2,\ldots,R$}
        \State $L_r \gets u_r-\ell_r$
        \State $s_r \gets \ell_r+2L_r/3$;\quad $\delta'_r \gets L_r/3$
        \State $k_r \gets R-r+1$;\quad $\epsilon_r \gets 1/(5k_r^2)$
        \State $b \gets \mathcal{D}_{\epsilon_r,s_r,\delta'_r,\gamma}(A,U)$
        \If{$b=1$}
            \State $(\ell_{r+1},u_{r+1}) \gets (s_r-\delta'_r,u_r)$
        \Else
            \State $(\ell_{r+1},u_{r+1}) \gets (\ell_r,s_r)$
        \EndIf
    \EndFor
    \State \Return $\hat\theta \gets (\ell_{R+1}+u_{R+1})/2$
\end{algorithmic}
\end{algorithm}

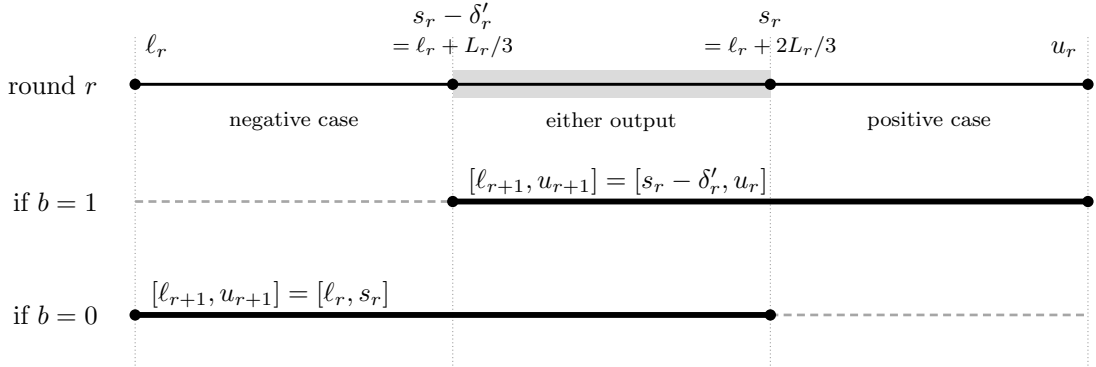
\begin{figure}[H]
\centering
\begin{tikzpicture}[
    x=0.9cm,
    y=1cm,
    current/.style={line width=1.1pt},
    kept/.style={line width=2.2pt},
    discarded/.style={line width=1pt,densely dashed,black!35},
    guide/.style={densely dotted,black!40},
    endpoint/.style={circle,fill=black,inner sep=1.5pt},
    intervalbrace/.style={
        decorate,
        decoration={brace,amplitude=4pt,mirror}
    },
    every node/.style={font=\small}
]

\def\xell{0}
\def\xupper{14}

\pgfmathsetmacro{\xcut}{\xell+(\xupper-\xell)/3}
\pgfmathsetmacro{\xs}{\xell+2*(\xupper-\xell)/3}

\pgfmathsetmacro{\xnegative}{(\xell+\xcut)/2}
\pgfmathsetmacro{\xeither}{(\xcut+\xs)/2}
\pgfmathsetmacro{\xpositive}{(\xs+\xupper)/2}

\foreach \x in {\xell,\xcut,\xs,\xupper}
    \draw[guide] (\x,-0.72) -- (\x,3.65);

\fill[black!14] (\xcut,2.82) rectangle (\xs,3.18);
\draw[current] (\xell,3) -- (\xupper,3);

\foreach \x in {\xell,\xcut,\xs,\xupper}
    \node[endpoint] at (\x,3) {};

\node[anchor=east] at (-0.4,3) {round $r$};

\node[anchor=south west] at (\xell,3.22)
    {$\ell_r$};

\node[anchor=south,align=center] at (\xcut,3.22)
    {$s_r-\delta'_r$\\[-1pt]
     \scriptsize $=\ell_r+L_r/3$};

\node[anchor=south,align=center] at (\xs,3.22)
    {$s_r$\\[-1pt]
     \scriptsize $=\ell_r+2L_r/3$};

\node[anchor=south east] at (\xupper,3.22)
    {$u_r$};

\node[anchor=north,align=center,font=\scriptsize]
    at (\xnegative,2.77)
    {negative case};

\node[anchor=north,align=center,font=\scriptsize]
    at (\xeither,2.77)
    {either output};

\node[anchor=north,align=center,font=\scriptsize]
    at (\xpositive,2.77)
    {positive case};

\draw[discarded] (\xell,1.45) -- (\xcut,1.45);
\draw[kept]      (\xcut,1.45) -- (\xupper,1.45);

\node[endpoint] at (\xcut,1.45) {};
\node[endpoint] at (\xupper,1.45) {};

\node[anchor=east] at (-0.4,1.45)
    {if $b=1$};

\node[anchor=west,xshift=2pt] at (\xcut,1.72)
    {$[\ell_{r+1},u_{r+1}]
      =[s_r-\delta'_r,u_r]$};


\draw[kept]      (\xell,-0.05) -- (\xs,-0.05);
\draw[discarded] (\xs,-0.05) -- (\xupper,-0.05);

\node[endpoint] at (\xell,-0.05) {};
\node[endpoint] at (\xs,-0.05) {};

\node[anchor=east] at (-0.4,-0.05)
    {if $b=0$};

\node[anchor=west,xshift=2pt] at (\xell,0.22)
    {$[\ell_{r+1},u_{r+1}]
      =[\ell_r,s_r]$};


\end{tikzpicture}
\caption{Updating the interval in one search round.}
\label{fig:interval-update}
\end{figure}

\begin{theorem}\label{thm:main}
    The algorithm in \Cref{alg:max-eigenphase-search} solves the max eigenphase estimation problem (\Cref{def:estimation}) with bounded error using:
    \begin{center}
        $\bigO{\frac{1}{\gamma} \log \frac{1}{\delta} \log\log\frac{1}{\delta}}$ controlled calls to $A$, and $A^\dagger$, $\qquad\bigO{\frac{1}{\gamma \delta}}$ controlled calls to $U$, and $U^\dagger$, 
    \end{center}
    and $\bigO{\frac{1}{\gamma \delta} \log \frac{1}{\gamma \delta}+\frac{1}{\gamma}\log N \log \frac{1}{\delta}\log\log\frac{1}{\delta}}$ one- and two-qubit gates and classical operations.
\end{theorem}

\begin{proof}
    The algorithm maintains a decreasing interval $[\ell_r,u_r]\subseteq[0,\pi/2]$ with the invariant $\theta_{\max}\in [\ell_r,u_r]$ holding except with small probability.

At each round, the size of the interval decreases by a factor $2/3$, so $L_r=\frac{\pi}{2}(\frac{2}{3})^{r-1}$. In particular, 
    $$L_{R+1}=\frac{\pi}{2}\left(\frac{2}{3}\right)^{R}\leq \frac{\pi}{2}\left(\frac{3}{2}\right)^{-\log_{3/2}(\pi/(4\delta))}=\frac{\pi}{2}\frac{4\delta}{\pi}=2\delta.$$
    Thus the center point, which we output, is within $\delta$ of $\theta_{\max}$, as long as $\theta_{\max}$ remains in the interval, which it does if every call to ${\cal D}$ is correct. As already argued the probability that any call to ${\cal D}$ errs is upper bounded by $1/3$. This establishes that the algorithm is correct with bounded error $1/3$.

    We now analyze the complexity of the algorithm. 
    The algorithm terminates after $R = \bigO{\log\frac{1}{\delta}}$ rounds.
    Each round of the computation uses one call to ${\cal D}$, which, by \Cref{cor:threshold-eigenphase-alg}, uses $O(\frac{1}{\delta_r'\gamma}\log\frac{1}{\epsilon_r})$ controlled calls to $U$ and $U^\dagger$, and $O(\frac{1}{\gamma}\log\frac{1}{\epsilon_r})$ controlled calls to $A$ and $A^\dagger$.
    Using $L_r=\frac{\pi}{2}(2/3)^{r-1}=L_R(3/2)^{R-r}$ and $L_R\geq \delta/3$, note that:
    \begin{equation}\label{eq:calls-to-U}\begin{split}
        \sum_{r=1}^{R}\frac{1}{\delta_r'}\log\frac{1}{\epsilon_r}
    &=\sum_{r=1}^{R}\frac{3}{L_r}\log(5(R-r+1)^2)
    =\bigO{\sum_{r=1}^R\frac{1}{L_R} (2/3)^{R-r}\log(R-r+1)}\\
    &=\bigO{\frac{1}{L_R}\sum_{r=1}^R(2/3)^r\log r}=\bigO{\frac{1}{L_R}}=\bigO{\frac{1}{\delta}}.
    \end{split}
    \end{equation}
    The total number of controlled calls to $U$ and $U^\dagger$ is then $\bigO{\frac{1}{\gamma\delta}}$, as claimed.

    Next, note that:
   \begin{equation}\label{eq:calls-to-A}
   \begin{split}
        \sum_{r=1}^R \log\frac{1}{\epsilon_r} &= \bigO{\sum_{r=1}^R \log(R-r+1)} = \bigO{ \sum_{r=1}^R \log(r)}\\
        &= \bigO{R \log R } = \bigO{ \log \frac{1}{\delta}\log\log\frac{1}{\delta}}.
    \end{split}
    \end{equation}
    The number of calls to $A$ and $A^\dagger$ is thus 
$\bigO{\frac{1}{\gamma} \log \frac{1}{\delta}\log\log\frac{1}{\delta}}$.
    By \cref{cor:threshold-eigenphase-alg}, in the $r$-th round, ${\cal D}$ uses up to $\bigO{(\frac{1}{\gamma \delta_r'}\log \frac{1}{\delta_r'\gamma} +\frac{1}{\gamma}\log N) \log \frac{1}{\eps_r}}$ elementary quantum gates, and this also upper bounds the classical operations in each loop iteration. Summing these costs gives:
    \begin{align*}
        \sum_{r=1}^R \bigO{\Big(\frac{1}{\gamma \delta_r'}\log \frac{1}{\delta_r'\gamma} +\frac{1}{\gamma}\log N\Big) \log \frac{1}{\eps_r}} &=  \bigO{\frac{1}{\gamma}\log \frac{1}{\delta \gamma} \sum_r \frac{1}{ \delta_r'}\log \frac{1}{\eps_r} } + \bigO{\frac{1}{\gamma}\log N \sum_{r=1}^R \log \frac{1}{\eps_r}}.
    \end{align*}
    The first term is $\bigO{\frac{1}{\delta\gamma}\log\frac{1}{\delta\gamma}}$ by \Cref{eq:calls-to-U}. The second term is $ \bigO{\frac{1}{\gamma}\log N \log \frac{1}{\delta}\log\log\frac{1}{\delta}}$, by \Cref{eq:calls-to-A}, which gives the claimed complexity.
\end{proof}

\bibliographystyle{alpha}
\bibliography{refs}

\end{document}